\newif\ifnotes
\documentclass[a4paper,UKenglish,cleveref,autoref,thm-restate]{lipics-v2021}

\usepackage[T1]{fontenc}

\usepackage{mathtools}
\usepackage{amsmath}
\usepackage{amsfonts}
\usepackage{amsthm}
\usepackage{url}
\usepackage{tikz-cd}  % not allowed with sigconf template
\usepackage{scalefnt} % not allowed with sigconf template
\usepackage{microtype}
\usepackage{mathpartir} % not allowed with sigconf template
\ifnotes
        \usepackage{todonotes} % not allowed with sigconf template
\else
        \usepackage[disable]{todonotes} % not allowed with sigconf template
\fi
\usepackage{cleveref}
\usepackage{xspace}
\usepackage{stmaryrd}
\newcommand{\MHELastName}{\texorpdfstring{Escard\'{o}}{Escardo}}

\newcommand{\MHEName}{%
  \texorpdfstring{Mart\'{i}n~H.~Escard\'{o}}{Martin\ H.\ Escardo}%
}
\newcommand{\BRPName}{Bruno~da~Rocha~Paiva}
\newcommand{\VRName}{Vincent~Rahli}
\newcommand{\ATName}{Ayberk~Tosun}

\newcommand{\IsDefinedToBe}{\vcentcolon\equiv}
\newcommand{\DeclareType}[2]{#1 : #2}
\newcommand{\NewDefinition}[2]{#1 \IsDefinedToBe #2}

\newcommand{\definiendum}[1]{\textbf{#1}}

\newcommand{\PowT}[1]{{#1}^{\mathsf{T}}}

\newcommand{\EqUpTo}[3]{#1 =_{#3} #2}

\newcommand{\AgdaLogo}{\raisebox{-0.2em}{\includegraphics[height=0.9em]{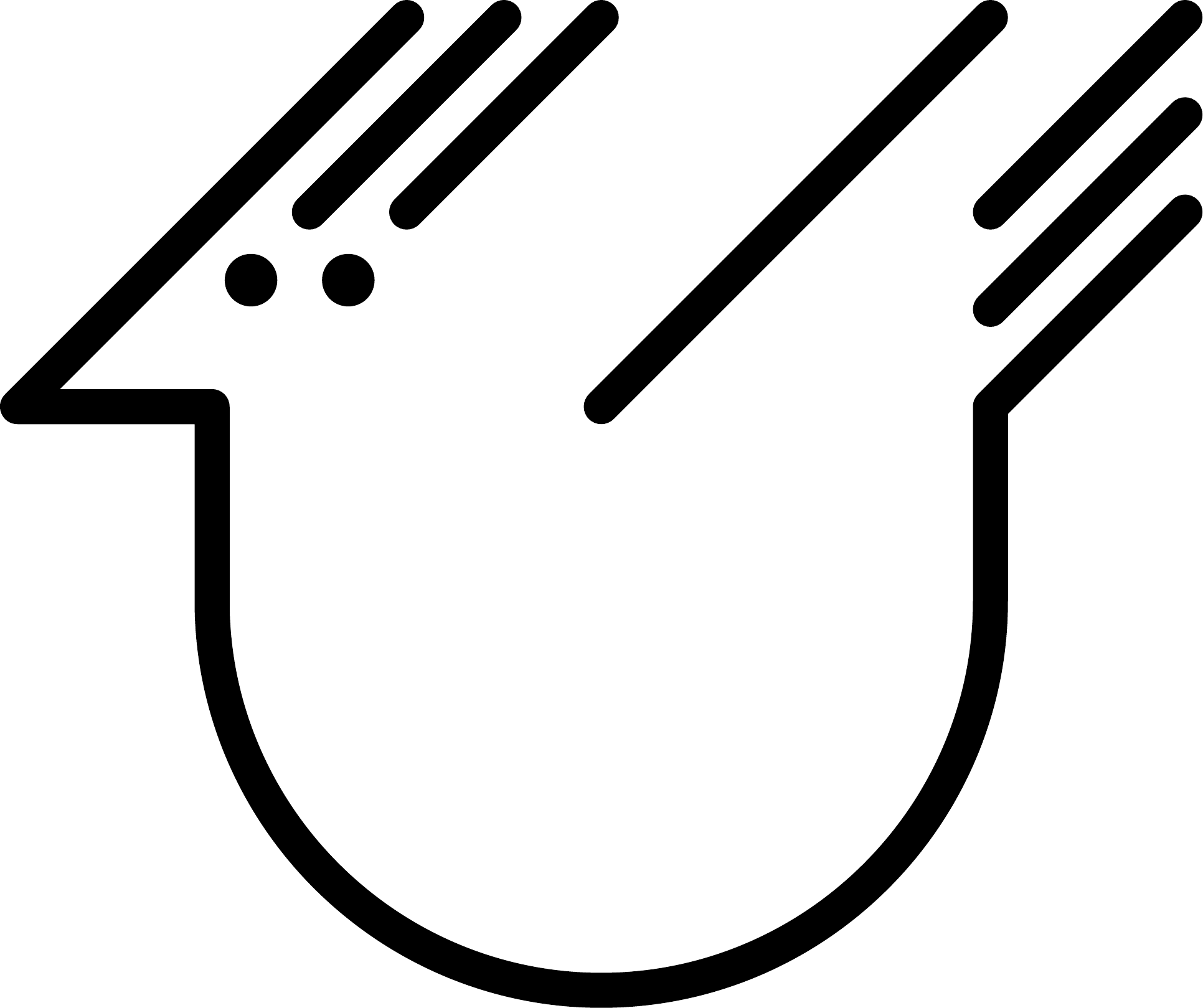}}}
\newcommand{\AgdaLink}[1]{%
  \href{%
    https://cs.bham.ac.uk/~mhe/InternalEffectfulForcing/EffectfulForcing.Internal.PaperIndex.html\##1%
  }{%
    \AgdaLogo{}%
  }\hspace{0.1341em}%
}

\newcommand{\PiTypeSymbol}{\prod}

\newcommand{\Fam}[2]{\{{#1}\}_{#2}}

\newcommand{\mlto}{\to}
\newcommand{\ArrowType}[2]{#1 \mlto #2}

\newcommand{\SigmaTypeSymbol}{\sum}

\newcommand{\HEEqualsSymbol}{\approx}
\newcommand{\HEEquals}[3]{{#1}\HEEqualsSymbol_{#3}{#2}}

\newcommand{\Fun}[3]{\lambda {#1 : #2}.\ #3}

\newcommand{\UFun}[2]{\lambda {#1}.\ #2}

\newcommand{\AppI}[2]{#1(#2)}
\newcommand{\AppII}[3]{#1(#2)(#3)}
\newcommand{\AppIII}[4]{#1(#2)(#3)(#4)}

\newcommand{\UU}{\mathsf{Type}}

\newcommand{\IdTypeSymbol}{=}
\newcommand{\IdTy}[2]{{#1}\IdTypeSymbol{#2}}
\newcommand{\IdType}[3]{{#1}\IdTypeSymbol_{#3}{#2}}

\newcommand{\DialogueSym}{\hyperref[def:dialogue-trees]%
  {\mathsf{Dial}}}
\newcommand{\Dialogue}[3]{\AppIII{\DialogueSym}{#1}{#2}{#3}}

\newcommand{\DialogueNNSym}{\hyperref[def:dialogue-trees]%
  {\mathcal{D}}}
\newcommand{\DialogueNN}[1]{\AppI{\DialogueNNSym}{#1}}

\newcommand{\KleisliSym}{\hyperref[def:kleisli_ext]%
  {\mathsf{kleisli{\mbox{-}}ext}}}
\newcommand{\Kleisli}[1]{\AppI{\KleisliSym}{#1}}

\newcommand{\GKleisliSym}{\hyperref[def:gen_kleisli_ext]%
  {\mathsf{Kleisli{\mbox{-}}ext}}}
\newcommand{\GKleisli}[1]{\GKleisliSym_{#1}}

\newcommand{\FunctorSym}{\hyperref[def:functor]%
  {\mathcal{D}\mathsf{\mbox{-}functor}}}
\newcommand{\Functor}[1]{\AppI{\FunctorSym}{#1}}

\newcommand{\Generic}{\hyperref[def:generic]%
  {\mathsf{generic}}}

\newcommand{\DialogueFSym}{\hyperref[def:dialogue_operator]%
  {\mathsf{dialogue}}}
\newcommand{\DialogueF}[2]{\AppII{\DialogueFSym}{#1}{#2}}

\newcommand{\DialogueTree}{\hyperref[def:dialogue_tree_operator]%
  {\mathsf{dialogue\mbox{-}tree}}}

\newcommand{\Nat}{\mathbb{N}}
\newcommand{\Zero}{0}
\newcommand{\One}{1}
\newcommand{\SuccSym}{1 +}
\newcommand{\Succ}[1]{\SuccSym{#1}}
\newcommand{\RecSym}{\mathsf{Natrec}}
\newcommand{\Rec}[3]{\AppIII{\RecSym}{#1}{#2}{#3}}

\newcommand{\Bool}{\mathbb{B}}

\newcommand{\MaxSym}{\mathsf{max}}
\newcommand{\Max}[2]{\AppII{\MaxSym}{#1}{#2}}

\newcommand{\MaxQuestionSym}{\hyperref[defn:max-question]%
  {\mathsf{max\mbox{-}q}}}
\newcommand{\MaxQuestion}[2]{\AppII{\MaxQuestionSym}{#1}{#2}}

\newcommand{\MaxBoolQuestionSym}{\hyperref[defn:max-bool-question]%
  {\mathsf{max\mbox{-}q_2}}}
\newcommand{\MaxBoolQuestion}[1]{\AppI{\MaxBoolQuestionSym}{#1}}

\newcommand{\ModulusSym}{\hyperref[defn:modulus]%
  {\mathsf{modulus}}}
\newcommand{\Modulus}[2]{\AppII{\ModulusSym}{#1}{#2}}

\newcommand{\ModulusUniSym}{\hyperref[def:uniform_modulus]%
  {\mathsf{modulus}_{2}}}
\newcommand{\ModulusUni}[1]{\AppI{\ModulusUniSym}{#1}}

\newcommand{\PruneSym}{\hyperref[def:pruning]%
  {\mathsf{prune}}}
\newcommand{\Prune}[1]{\AppI{\PruneSym}{#1}}

\newcommand{\EmbedSym}{\hyperref[def:pruning]%
  {\mathsf{embed}}_{\Bool}%
}
\newcommand{\Embed}[1]{\AppI{\EmbedSym}{#1}}

\newcommand{\EmbedCBSym}{%
  \mathsf{embed}_{C}%
}
\newcommand{\EmbedCB}[1]{%
  \AppI{\EmbedCBSym}{#1}%
}

\newcommand{\ChurchIntNNSym}[1]{\hyperref[def:internal_dialogue_trees]%
  {\mathcal{D}^{\mathsf{T}}_{#1}}}
\newcommand{\ChurchIntNN}[2]{\AppI{\ChurchIntNNSym{#2}}{#1}}

\newcommand{\LeafIntSym}[1]{\hyperref[def:internal_dialogue_trees]%
  {\eta^{\mathsf{T}}_{#1}}}
\newcommand{\LeafInt}[2]{\TAppI{\LeafIntSym{#2}}{#1}}

\newcommand{\BranchIntSym}[1]{\hyperref[def:internal_dialogue_trees]%
  {\beta^{\mathsf{T}}_{#1}}}

\newcommand{\EncodeExtSym}[1]{\hyperref[def:encode]%
  {\mathsf{encode}_{#1}}}
\newcommand{\EncodeExt}[2]{\EncodeExtSym{#2}({#1})}

\newcommand{\TSub}[2]{({#1})\{{#2}\}}

\newcommand{\TTypes}{\hyperref[def:systemT-syntax]%
  {\mathsf{Type}^\mathsf{T}}}

\newcommand{\TCtxs}{\hyperref[def:systemT-syntax]%
  {\mathsf{Ctx}^\mathsf{T}}}

\newcommand{\TTerms}[2]{\hyperref[def:systemT-syntax]%
  {\mathsf{Term}^\mathsf{T}({#1},{#2})}}

\newcommand{\TCTerms}[1]{\hyperref[def:systemT-syntax]%
  {\mathsf{Term}^\mathsf{T}_0({#1})}}

\newcommand{\TNat}{\hyperref[def:systemT-syntax]%
  {\iota}}

\newcommand{\tto}{\hyperref[def:systemT-syntax]%
  {\Rightarrow}}
\newcommand{\TArrowType}[2]{#1 \tto #2}

\newcommand{\TZero}{\hyperref[def:systemT-syntax]%
  {\mathtt{Zero}}}

\newcommand{\TSuccSym}{\hyperref[def:systemT-syntax]%
  {\mathtt{Succ}}}
\newcommand{\TSucc}[1]{\TAppI{\TSuccSym}{#1}}

\newcommand{\TRecSym}{\hyperref[def:systemT-syntax]%
  {\mathtt{Rec}}}
\newcommand{\TRec}[4]{\TAppI{\TRecSym_{#1}}{#2,#3,#4}}

\newcommand{\lambdabar}{{\mkern0.75mu\mathchar '26\mkern -9.75mu\lambda}}
\newcommand{\TLambdaSym}{\hyperref[def:systemT-syntax]%
  {\lambdabar}}
\newcommand{\TFun}[3]{\TLambdaSym {(#1 : #2)}.\ #3}
\newcommand{\TUFun}[2]{\TLambdaSym {#1}.\ #2}

\newcommand{\TAppI}[2]{#1(#2)}
\newcommand{\TAppII}[3]{#1(#2)(#3)}
\newcommand{\TAppIII}[4]{#1(#2)(#3)(#4)}

\newcommand{\TNumeral}[1]{\underline{#1}}

\newcommand{\TKleisliSym}[1]{\hyperref[def:internal_kleisli]%
  {\mathtt{kleisli\mbox{-}ext}^{\mathsf{T}}_{#1}}}
\newcommand{\TKleisli}[2]{\TAppI{\TKleisliSym{#2}}{#1}}

\newcommand{\TGKleisliSym}[1]{\hyperref[def:internal_gkleisli]%
  {\mathtt{Kleisli\mbox{-}ext}^{\mathsf{T}}_{#1}}}
\newcommand{\TGKleisli}[2]{\TGKleisliSym{{#1},{#2}}}

\newcommand{\TFunctorSym}[1]{\hyperref[def:internal_functor]%
  {{\mathcal{D}\mathtt{\mbox{-}functor}}^{\mathsf{T}}_{#1}}}
\newcommand{\TFunctor}[2]{\TAppI{\TFunctorSym{#2}}{#1}}

\newcommand{\TGeneric}[1]{\hyperref[def:internal_generic]%
  {\mathtt{generic}^{\mathsf{T}}_{#1}}}

\newcommand{\TDialogueFSym}{\hyperref[defn:dialogue_internal]%
  {\mathtt{dialogue}^{\mathsf{T}}}}

\newcommand{\TDialogueTreeSym}[1]{\hyperref[def:internal_dialogue_tree_operator]%
  {\mathtt{dialogue\mbox{-}tree}^{\mathsf{T}}_{#1}}}
\newcommand{\TDialogueTree}[2]{\AppI{\TDialogueTreeSym{#2}}{#1}}

\newcommand{\empctx}{\diamond}
\newcommand{\typed}[3]{{#2}:\TTerms{#1}{#3}}
\newcommand{\ctyped}[2]{{#1}:\TCTerms{#2}}

\newcommand{\TMaxSym}{\PowT{\mathtt{max}}}
\newcommand{\TMax}[2]{\AppII{\PowT{\mathtt{max}}}{#1}{#2}}

\newcommand{\TMaxQuestionSym}{\hyperref[defn:max-question-int]%
  {\PowT{\mathtt{max\mbox{-}q}}}}
\newcommand{\TMaxQuestion}[2]{\AppII{\TMaxQuestionSym}{#1}{#2}}

\newcommand{\MaxBoolQuestionIntSym}{\hyperref[def:internal_max_bool_question]%
  {\PowT{\mathtt{max\mbox{-}q}}_2}}
\newcommand{\MaxBoolQuestionInt}[1]{\TAppI{\MaxBoolQuestionIntSym}{#1}}

\newcommand{\ModulusIntSym}{\hyperref[defn:modulus]%
  {\PowT{\mathtt{modulus}}}}

\newcommand{\ModulusUniIntSym}{\hyperref[def:uniform_modulus]%
  {\PowT{\mathtt{modulus}}_{2}}}
\newcommand{\ModulusUniInt}[1]{\AppI{\ModulusUniIntSym}{#1}}

\newcommand{\typestodial}[1]{\llbracket{#1}\rrbracket_{\mathcal{D}}}
\newcommand{\ctxstodial}[1]{\llbracket{#1}\rrbracket_{\mathcal{D}}}
\newcommand{\termstodial}[2]{\llbracket{#1}\rrbracket_{\mathcal{D}}^{#2}}

\newcommand{\typestoset}[1]{\llbracket{#1}\rrbracket_\set}
\newcommand{\ctxstoset}[1]{\llbracket{#1}\rrbracket_\set}
\newcommand{\termstoset}[2]{\llbracket{#1}\rrbracket_\set^{#2}}

\newcommand{\typestoint}[2]{\llbracket{#1}\rrbracket_{\mathcal{D}^\mathsf{T}}^{#2}}
\newcommand{\ctxstoint}[2]{\llbracket{#1}\rrbracket_{\mathcal{D}^\mathsf{T}}^{#2}}
\newcommand{\termstoint}[2]{\llbracket{#1}\rrbracket_{\mathcal{D}^\mathsf{T}}^{#2}}

\newcommand{\Rdial}[4]{{#1}\vDash{#2}\approx{#3}\in{#4}}

\newcommand{\Rnorm}[3]{{#1}\sim{#2}\in{#3}}

\newcommand{\nat}{\mathbb{N}}

\newcommand{\LeafSym}{\eta}
\newcommand{\Leaf}[1]{\AppI{\LeafSym}{#1}}
\newcommand{\BranchSym}{\beta}
\newcommand{\Branch}[2]{\AppII{\BranchSym}{#1}{#2}}

\newcommand{\set}{\mathbf{Set}}

\newcommand{\VSystemT}{System~T\xspace}

\newcommand{\VAgda}{\textsc{Agda}\xspace}

\newcommand{\VMLTT}{\textsf{MLTT}\xspace}

\definecolor{Melon}{HTML}{F89E7B}

\Crefname{equation}{Eq.}{Eqs.}
\Crefname{figure}{Fig.}{Figs.}
\Crefname{tabular}{Tab.}{Tabs.}
\Crefname{section}{Sec.}{Secs.}
\Crefname{definition}{Def.}{Defs.}
\Crefname{definition}{Def.}{Defs.}
\Crefname{lemma}{Lem.}{Lems.}
\Crefname{theorem}{Thm.}{Thms.}
\Crefname{theorem}{Thm.}{Thms.}
\Crefname{paragraph}{Sec.}{Secs.}
\Crefname{appendix}{Appx.}{Appxs.}
\Crefname{corollary}{Cor.}{Cors.}
\Crefname{example}{Ex.}{Exs.}
\Crefname{example}{Ex.}{Exs.}
\Crefname{proposition}{Prop.}{Props.}

\title{Internal Effectful Forcing in \VSystemT{}}

\author%
    {\MHEName}
    {University of Birmingham}
    {m.escardo@bham.ac.uk}
    {0000-0002-4091-6334}
    {}

\author%
    {\BRPName}
    {University of Birmingham}
    {bmd202@student.bham.ac.uk}
    {0000-0002-2205-8815}
    {}

\author%
    {\VRName}
    {University of Birmingham}
    {v.rahli@bham.ac.uk}
    {0000-0002-5914-8224}
    {}

\author%
    {\ATName}
    {University of Birmingham}
    {a.tosun@pgr.bham.ac.uk}
    {0000-0002-0190-3020}
    {}

\authorrunning{\MHELastName{} et al.} %TODO mandatory. First: Use abbreviated first/middle names. Second (only in severe cases): Use first author plus 'et al.'

\Copyright{\MHEName, \BRPName, \VRName,  and \ATName} %TODO mandatory, please use full first names. LIPIcs license is "CC-BY";  http://creativecommons.org/licenses/by/3.0/



\category{} %optional, e.g. invited paper

\relatedversion{} %optional, e.g. full version hosted on arXiv, HAL, or other respository/website
\nolinenumbers %uncomment to disable line numbering

\EventEditors{Maribel Fern\'{a}ndez}
\EventNoEds{1}
\EventLongTitle{10th International Conference on Formal Structures for Computation and Deduction (FSCD 2025)}
\EventShortTitle{FSCD 2025}
\EventAcronym{FSCD}
\EventYear{2025}
\EventDate{July 14--20, 2025}
\EventLocation{Birmingham, UK}
\EventLogo{}
\SeriesVolume{337}
\ArticleNo{16}
\begin{document}

\maketitle

%% Very rudimentary abstract.
\begin{abstract}
  The \emph{effectful forcing} technique allows one to show that the
  denotation of a closed \VSystemT{} term of type $(\iota \Rightarrow \iota) \Rightarrow \iota$ in the set-theoretical model is
  a continuous function $(\Nat \to \Nat) \to \Nat$. For this purpose, an alternative dialogue-tree semantics
  is defined and related to the set-theoretical semantics by a logical
  relation. In this paper, we apply effectful forcing to show that the
  dialogue tree of a \VSystemT{} term is itself \VSystemT{}-definable,
  using the Church encoding of trees. % This seems to be constructively
  % stronger than the previously known result that the modulus of
  % continuity of a \VSystemT{}-term is itself \VSystemT{}-definable. % better say this in the introduction.
\end{abstract}

%% The effectful forcing technique allows one to show that the denotation
%% of a System T term in the set-theoretical model is continuous. For
%% this purpose, an alternative dialogue-tree semantics is defined and
%% related to the set-theoretical semantics by a logical relation. In
%% this paper, we apply effectful forcing to show that the dialogue tree
%% of a System T term is itself System T-definable, using Church encoding
%% of trees. This is constructively stronger than the previously known
%% result that the modulus of continuity of a System T-term is itself
%% System T-definable.


%% \ccsdesc[500]{Do Not Use This Code~Generate the Correct Terms for Your Paper}
%% \ccsdesc[300]{Do Not Use This Code~Generate the Correct Terms for Your Paper}
%% \ccsdesc{Do Not Use This Code~Generate the Correct Terms for Your Paper}
%% \ccsdesc[100]{Do Not Use This Code~Generate the Correct Terms for Your Paper}


\todo[inline]{be consistent about using model vs translation vs interpretation.
maybe use model for actual models and translation for non models?}

%% BRP -- Introduction
\section{Introduction}%
\label{sec:intro}

It is well known that the \VSystemT{}-definable functions
$(\Nat \to \Nat) \to \Nat$ are continuous and that, moreover, their
moduli of continuity are themselves
\VSystemT{}-definable~\cite{Troelstra:1973}.  Effectful
forcing~\cite{mhe-effectful-forcing} was generalised by
Xu~\cite{xu:2020} to give an alternative proof of this fact. Effectful
forcing gives a dialogue-tree semantics, and relates it to the
set-theoretical semantics by a logical relation. In this paper, we
strengthen this by showing that the dialogue trees are themselves
\VSystemT{}-definable, using Church encoding. Dialogue trees are
equivalent variants of Brouwer
trees~\cite{Escardo+Oliva:diag2brouwer:2017,Sterling:jfp:2021}

From a constructive point of view, dialogue trees give more
information than moduli of continuity. Given a dialogue tree, it is
possible to derive a modulus of continuity, but the converse is only
known to be possible in the presence of additional assumptions. For
example, Ghani, Hancock and
Pattison~\cite{Ghani+Hancock+Pattinson:lmcs:2009} show that if a
function doesn't have a Brouwer tree, then it isn't continuous, using
dependent choice, while Capretta and
Uustalu~\cite{Capretta+Uustalu:fossacs:2016} use the assumption of bar
induction to show that every function with a \emph{stable} modulus of
continuity has a Brouwer tree. Our result does not assume dependent
choice or bar induction, and, moreover, establishes the
\emph{\VSystemT{}-definability} of a dialogue tree of any closed term
of type~$(\iota \Rightarrow \iota) \Rightarrow \iota$.

\emph{Related work.}
Continuity is a key concept in mathematics, and in particular in constructive
mathematics where it is often accepted that all real-valued functions on the
unit interval are uniformly continuous.
Prominent work in the area was done by Brouwer, who gave the first argument for
the previous result, relying on his \emph{continuity principle for numbers},
which states that all functions on the Baire space are continuous, as well as
the Fan Theorem~\cite{Troelstra+VanDalen:1988,Dummett:1977}.
Brouwer's continuity principle said that the values of a function
\(\DeclareType{F}{\ArrowType{(\ArrowType{\Nat}{\Nat})}{\Nat}}\) on the Baire
space could only rely on a finite amount of its inputs.
More precisely, \(F\) was said to be continuous if for all
\(\DeclareType{\alpha}{\ArrowType{\Nat}{\Nat}}\), there existed some
\(\DeclareType{n}{\Nat}\) such that for any inputs
\(\DeclareType{\beta}{\ArrowType{\Nat}{\Nat}}\) which agreed on the first \(n\)
entries with \(\alpha\) we would have \(\AppI{F}{\alpha}=\AppI{F}{\beta}\).
When the value of \(n\) could be picked independently of \(\alpha\), then
\(F\) was said to be uniformly continuous.
Since then, more fine-grained ways of capturing continuity information have been
devised, such as using trees to keep track of the specific entries of \(\alpha\)
that \(\AppI{F}{\alpha}\) might depend on.
These ideas were used by Kleene in his study of recursive functions of
higher types~\cite{Kleene:1978}, by Brouwer in his study of his bar
theorem~\cite{Troelstra+VanDalen:1988}, and many others, including~\cite{Ghani+Hancock+Pattinson:cmcs:2006,Ghani+Hancock+Pattinson:lmcs:2009,Ghani+Hancock+Pattinson:mfps:2009,mhe-effectful-forcing,Capretta+Uustalu:fossacs:2016,Sterling:jfp:2021,Baillon+Mahboubi+Pedrot:csl:2022,Cohen+Paiva+Rahli+Tosun:mfcs:2023,Baillon:phd:2023}.

Starting with Troelstra's work~\cite[p.158]{Troelstra:1973} for
\(\text{N-HA}^{\omega}\), the definable functions of various other systems have been shown to be continuous.
Among these we have:
System~T~\cite{mhe-effectful-forcing},
MLTT~\cite{Coquand+Jaber:2010, Coquand+Jaber:2012, Xu:phd:2015},
CTT~\cite{Rahli+Bickford:cpp:2016},
BTT~\cite{Baillon+Mahboubi+Pedrot:csl:2022}, and
\(\text{TT}^{\Box}_{\mathcal{C}}\)~\cite{Cohen+Paiva+Rahli+Tosun:mfcs:2023}.
However, the existence of a modulus of continuity operator tends to
be inconsistent with function extensionality as first shown
for \(\text{HA}_{\mathit{NB}}^{\omega}\) by Kreisel~\cite[p.154]{Kreisel:1962}
and later
for \(\text{N-HA}^{\omega}\) by Troelstra~\cite[Thm.3.1(IIA)]{Troelstra:1977b}.
In the case of MLTT we don't even need function extensionality
% with reduction under \(\lambda\) already being enough
to get an inconsistency as shown by
Escard\'o and Xu~\cite{Escardo+Xu:2015,Xu:phd:2015}.
%
% As a result, these results are either proved externally or the existence of a
% modulus of continuity is relaxed.
%
% In this paper we continue this line of work to further study the continuity of
% \VSystemT{} functionals, building upon the technique of effectful
% forcing~\cite{mhe-effectful-forcing}.
%
%
Further related work is discussed in \Cref{sec:conclusion}.

% The internalisation carried out in this work allows us to bridge this gap,
% recovering the stronger continuity results.
% %
% It is worth noting that while Brouwer used the Fan Theorem to derive uniform
% continuity from his continuity principle on numbers, the proof
% in~\cite{mhe-effectful-forcing} as well as the proof in the present paper, do
% not rely on the Fan Theorem or on classical principles, taking advantage instead
% of an external tabulation of functions into an internal inductive type.

\emph{Main contributions.}
(1) We strengthen the above work to show that not only the modulus of
continuity of a \VSystemT-definable function is itself
\VSystemT-definable, but also its dialogue tree is
\VSystemT-definable, where we implemented trees in \VSystemT using
Church encoding.
(2) We prove the
correctness of this translation with a logical relation.
(3) We show how to
compute moduli of continuity and uniform continuity internally, and prove the
correctness of this construction.
(4) The original presentation of effectful forcing~\cite{mhe-effectful-forcing} extends \VSystemT with an oracle, and here we show that this extension is not necessary and we can work directly with \VSystemT without oracles.

\emph{Organisation.}
\Cref{sec:system_T} recalls \VSystemT{}'s syntax as well as its set
semantics~\(\termstoset{-}{}\).
\Cref{sec:effectful_forcing} presents a simplified version of
effectful forcing using the inductive dialogue
translation~\(\termstodial{-}{}\) and a modified version of the
logical relation~\(\Rdial{\alpha}{x}{y}{\sigma}\).
\Cref{sec:translation} presents the Church encodings of dialogue
trees, the resulting internal dialogue
translation~\(\termstoint{-}{}\) and the logical
relation~\(\Rnorm{x}{y}{\sigma}\) linking inductive dialogue trees and
internal Church encoded dialogue trees.
Finally, \Cref{sec:computing-moduli,sec:uniform} explain how to
compute moduli of continuity and uniform moduli of continuity inside
\VSystemT{} from these dialogue trees. These translations and logical
relations are summarised in the following diagram:
% For a summary
% of the mentioned translations and logical relations linking them see
%\Cref{fig:roadmap}.

%\begin{figure}[!h]
\begin{center}
\begin{tikzcd}
  & & {\TCTerms{\TNat}} \arrow[dd, "\termstodial{-}{}"] \arrow[lldd, "\termstoset{-}{}"', out=180, in=90, looseness=0.8] \arrow[rrdd, "\termstoint{-}{}", out=0, in=90, looseness=0.8] & & \\
  & & & & \\
  \Nat \arrow[rr, "\Rdial{\alpha\,}{\,\cdot\,}{\,\cdot\,}{\TNat}"', Leftrightarrow] & & \DialogueNN{\Nat} \arrow[rr,
  "\Rnorm{\,\cdot\,}{\,\cdot\,}{\TNat}"', Leftrightarrow] & &
  {\TCTerms{\ChurchIntNN{\TNat}{A}}}
\end{tikzcd}
\end{center}
%\caption{Translations featured in this paper and the logical relations linking them.}%
%\label{fig:roadmap}
%\end{figure}

%\subsection{Our Metatheory}%
\label{sub:metatheory}

\emph{Metatheory.} Our metatheory is a spartan version of \VMLTT{}
featuring \(\PiTypeSymbol\)-types, \(\SigmaTypeSymbol\)-types,
inductive types, and a universe. All of our reasoning is constructive
and we don't rely on function extensionality. We will fix some
notation for the rest of the paper: we denote the type of natural
numbers by \(\Nat\), its zero element by \(\Zero\), given a natural
number \(n\) we write \(\Succ{n}\) for its successor, and we let
\(\DeclareType{\RecSym}{\ArrowType{(\ArrowType{\Nat}{\ArrowType{X}{X}})}{\ArrowType{X}{\ArrowType{\Nat}{\Nat}}}}\)
be its recursor.

\emph{Formalisation.} The results of this paper have been formalised in \VAgda{}
and we have tried to keep the presentation faithful to the formalisation~\cite{EPRT}.
The main difference is the use of variables for clarity, while the
formalisation uses de Bruijn indices for \VSystemT{}.

\section{\VSystemT}%
\label{sec:system_T}

\VSystemT{} features three notions: contexts, types and terms. It has a single
base type \(\TNat\) of natural numbers as well as a type of functions
\(\TArrowType{\sigma}{\tau}\). Contexts are lists of paired variables and types.
Terms are built up inductively from variables, a zero constant \(\TZero\), a
successor operation \(\TSuccSym\), a recursion operator \(\TRecSym\), lambda
abstraction and function application. More formally, we have the following
definitions.

\begin{figure}
  \begin{mathpar}
    \small

    \boxed{\TTypes}

    \hspace{4em}

    \boxed{\TCtxs}

    \hspace{1em}\\

    \inferrule{  }{\TNat : \TTypes}

    \inferrule{ \sigma : \TTypes \\ \tau : \TTypes}{\TArrowType{\sigma}{\tau} : \TTypes}

    \inferrule{  }{\empctx : \TCtxs}

    \inferrule{\Gamma : \TCtxs \\ (x:\sigma)\notin\Gamma}{\Gamma,(x:\sigma) : \TCtxs}\\

    \boxed{\TTerms{\Gamma}{\sigma}}

    \inferrule{(x : \sigma) \in \Gamma}{\typed{\Gamma}{x}{\sigma}}

    \inferrule{ }{\typed{\Gamma}{\TZero}{\TNat}}

    \inferrule{\typed{\Gamma}{t}{\TNat}}{\typed{\Gamma}{\TSucc{t}}{\TNat}}

    \inferrule{
      \typed{\Gamma}{t}{\TArrowType{\TNat}{\TArrowType{\sigma}{\sigma}}} \\
      \typed{\Gamma}{p}{\sigma} \\
      \typed{\Gamma}{q}{\TNat}
    }{\typed{\Gamma}{\TRec{\sigma}{t}{p}{q}}{\sigma}}

    \inferrule{\typed{\Gamma, x:\sigma}{t}{\tau}}{\typed{\Gamma}{\TFun{x}{\sigma}{t}}{\TArrowType{\sigma}{\tau}}}

    \inferrule{\typed{\Gamma}{t}{\TArrowType{\sigma}{\tau}} \\ \typed{\Gamma}{p}{\sigma}}{\typed{\Gamma}{\TAppI{t}{p}}{\tau}}
  \end{mathpar}
\caption{The syntax of intrinsically typed \VSystemT.}%
\label{fig:systemT-syntax}
\end{figure}

\begin{definition}[\AgdaLink{Definition-1}]\label{def:systemT-syntax}
  The types
  (1)~of \VSystemT{} types \(\TTypes\);
  (2)~of \VSystemT{} contexts \(\TCtxs\);
  and (3)~given a \VSystemT{} context \(\Gamma : \TCtxs\) and
  a \VSystemT{} type \(\sigma : \TTypes\), of \VSystemT{} terms of type \(\sigma\) in context
  \(\Gamma\), denoted \(\TTerms{\Gamma}{\sigma}\),
  are defined in \Cref{fig:systemT-syntax}.
  We fix the shorthand \(\TCTerms{\sigma}\) for \(\TTerms{\empctx}{\sigma}\).
\end{definition}

%% Finally we define the terms of \VSystemT{}. By defining the syntax of terms at
%% the same time as the typing rules, we make it impossible to write ill-typed
%% terms, so from now on whenever we mention a \VSystemT{} term it will be
%% implicit that it is also well-typed.

Given contexts \(\Gamma, \Sigma : \TCtxs\), a substitution from \(\Gamma\)
to \(\Sigma\) is an assignment of terms \(\typed{\Gamma}{t}{\sigma}\)
to each variable \((x:\sigma)\in\Delta\). Though we don't formally define it here,
we assume we have a capture-free substitution of \VSystemT{} terms.
Given a term \(\typed{\Gamma}{t}{\sigma}\)
and substitution \(\rho\) from \(\Delta\) to \(\Gamma\), we denote
the substituted term by \(\typed{\Delta}{\TSub{t}{\rho}}{\sigma}\).

At this point one would usually introduce the reduction rules associated with
\VSystemT{}, or the corresponding notion of conversion of terms. As we will see
later, the dialogue semantics of \VSystemT{} do not respect conversion of terms
in general, hence we omit these rules. With that said though, we now introduce
the set model~\cite{Troelstra:1973} of \VSystemT{}, sometimes also called the
functional model~\cite{Avigad+Feferman:dialectica:1999}. In this model we
interpret \(\TNat\) as the metatheoretic natural numbers \(\Nat\) and function
types \(\tto\) as the metatheoretic function space \(\mlto\). A
context~\(\Gamma\) is then interpreted as the type of assignments from the
context variables to elements of the translated types. A term
\(\typed{\Gamma}{t}{\sigma}\) is interpreted as a metatheoretic function from
the interpretation of~\(\Gamma\) to the interpretation of~\(\sigma\).

\begin{figure}
  \small
 \centering
 \begin{tabular}{m{0.38\textwidth}r}
   \(
   \begin{aligned}
     \DeclareType{\typestoset{-}&}{\ArrowType{\TTypes}{\UU}}\\
     \NewDefinition{\typestoset{\TNat}&}{\nat}\\
     \NewDefinition{\typestoset{\TArrowType{\sigma}{\tau}}&}{\ArrowType{\typestoset{\sigma}}{\typestoset{\tau}}} \\
\\
   \end{aligned}
   \)
        \(
        \begin{aligned}
          \DeclareType{\ctxstoset{-}&}{\ArrowType{\TCtxs}{\UU}}\\
          \NewDefinition{\ctxstoset{\Gamma}&}{\ArrowType{(x:\sigma)\in\Gamma}{\typestoset{\sigma}}}
        \end{aligned}
        \) &
   \(
   \begin{aligned}
     \DeclareType%
     {\termstoset{-}{-}&}%
     {\ArrowType{\TTerms{\Gamma}{\sigma}}{%
      \ArrowType{\typestoset{\Gamma}}{\typestoset{\sigma}}
     }}\\
     \NewDefinition{\termstoset{x}{\gamma}&}%
       {\AppI{\gamma}{x}}\\
     \NewDefinition{\termstoset{\TZero}{\gamma}&}%
       {\Zero}\\
     \NewDefinition{\termstoset{\TSucc{t}}{\gamma}&}%
       {\Succ{\termstoset{t}{\gamma}}}\\
     \NewDefinition{\termstoset{\TRec{\sigma}{t_1}{t_2}{t_3}}{\gamma}&}%
       {\Rec{\termstoset{t_1}{\gamma}}{\termstoset{t_2}{\gamma}}{\termstoset{t_3}{\gamma}}}\\
     \NewDefinition{\termstoset{\TFun{x}{\sigma}{t}}{\gamma}&}%
       {\Fun{x^\prime}{\typestoset{\sigma}}{\termstoset{t}{\gamma,(x \mapsto x^\prime)}}}\\
     \NewDefinition{\termstoset{\TAppI{t_1}{t_2}}{\gamma}&}%
       {\AppI{\termstoset{t_1}{\gamma}}{\termstoset{t_2}{\gamma}}}
   \end{aligned}
   \)
 \end{tabular}
\caption{Set model of \VSystemT.}%
\label{fig:set-model}
\end{figure}

\begin{definition}[\AgdaLink{Definition-2a}]\label{def:set-model}
  The \definiendum{set interpretation} of types, contexts and terms of \VSystemT
  is defined in \Cref{fig:set-model}. When defining the interpretation of terms
  we let \(\gamma\) range over~\(\ctxstoset{\Gamma}\). We fix the notation
  \(\gamma,(x \mapsto x^\prime)\) to extend an assignment of
  \(\ctxstoset{\Gamma}\) to an assignment of~\(\ctxstoset{\Gamma,(x:\sigma)}\).
\end{definition}

Each natural number can be encoded as a \VSystemT term in the standard manner.

\begin{definition}[\AgdaLink{Definition-3}]
  Given a natural number \(n\) we may define a \VSystemT{} term
  \(\ctyped{\TNumeral{n}}{\TNat}\) by induction on \(n\) according to the rules
  \(\NewDefinition{\TNumeral{\Zero}}{\TZero}\) and
  \(\NewDefinition{\TNumeral{\Succ{n}}}{\TSucc{\TNumeral{n}}}\).
\end{definition}

\begin{proposition}[\AgdaLink{Proposition-4}]
  For all \(\DeclareType{n}{\Nat}\), we have
  \(\IdType{n}{\termstoset{\TNumeral{n}}{}}{}\).
\end{proposition}

%%% Local Variables:
%%% mode: latexmk
%%% TeX-master: "../internal-effectful-forcing"
%%% End:

\section{Oracle-less Effectful Forcing}%
\label{sec:effectful_forcing}

The original presentation of effectful forcing~\cite{mhe-effectful-forcing}
works with \VSystemT extended with an oracle. Here we show that this extension
is not necessary, and we can work directly with \VSystemT.

% Our presentation differs
% slightly from~ as we do not extend \VSystemT{}
% with an oracle. %, hence the moniker \emph{oracless}.

\subsection{Dialogue Trees and Continuity}%
\label{sub:dialogue_trees}

Effectful forcing starts by noting that a function
\(f : \ArrowType{(\ArrowType{I}{O})}{X}\) can be thought of as an effectful term
\(t_{f} : X\) with access to an oracle \(\alpha : \ArrowType{I}{O}\). The use of
an oracle can be seen as an algebraic effect with the operation
\(\texttt{question}_{X}\;(i : I)\;(k : O \to X) : X\), which prompts the oracle
with question~\(i\) and passes the oracle's answer to the continuation~\(k\)
eventually producing an element of~\(X\). Looking for the monad corresponding to
oracle computations we find dialogue trees.

\begin{definition}[\AgdaLink{Definition-5}]
  \label{def:dialogue-trees}
  The inductive type \(\Dialogue{I}{O}{X}\) of \definiendum{\((I,O,X)\)-dialogue
    trees} is defined by
  \begin{mathpar}
    \inferrule{x:X}{\Leaf{x}:\Dialogue{I}{O}{X}}

    \inferrule{\phi : \ArrowType{O}{\Dialogue{I}{O}{X}} \\
      i:I}{\Branch{\phi}{i}:\Dialogue{I}{O}{X}}
  \end{mathpar}
  In the case where the oracle input and output are \(\Nat\), we write
  \(\DialogueNN{X}\) for \(\Dialogue{\Nat}{\Nat}{X}\).
\end{definition}

Justifying this view, each \((I,O,X)\)-dialogue tree encodes a function of
type~\(\ArrowType{(\ArrowType{I}{O})}{X}\).

\begin{definition}[\AgdaLink{Definition-6}]\label{def:dialogue_operator}
  The \definiendum{dialogue} between a dialogue tree and an oracle is given by
  \begin{gather*}
    \DeclareType{\DialogueFSym}{\ArrowType{\Dialogue{I}{O}{X}}{\ArrowType{(\ArrowType{I}{O})}{X}}} \\
    \NewDefinition{\DialogueF{\Leaf{x}}{\alpha}}{x} \\
    \NewDefinition{\DialogueF{\Branch{\phi}{i}}{\alpha}}{\DialogueF{\AppI{\phi}{\AppI{\alpha}{i}}}{\alpha}}
  \end{gather*}
\end{definition}

\begin{definition}[\AgdaLink{Definition-7a}]\label{defn:continuity}
  Given a function \(\DeclareType{f}{\ArrowType{(\ArrowType{I}{O})}{X}}\)
  we say it is:
  \begin{itemize}
    \item \definiendum{dialogue continuous} if
    \(
    \exists\,d:\Dialogue{I}{O}{X},
    \forall\,\alpha:\ArrowType{I}{O},
    \IdType{\AppI{f}{\alpha}}{\DialogueF{d}{\alpha}}{}
    \).
  \end{itemize}
  When the oracle input is \(I=\Nat\), so we have
  \(\DeclareType{f}{\ArrowType{(\ArrowType{\Nat}{O})}{X}}\), we also say \(f\)
  is
  \begin{itemize}
    \item \definiendum{continuous} if
    \(
    \forall\,\alpha:\ArrowType{\Nat}{O},
    \exists\,n:\Nat,
    \forall\,\beta:\ArrowType{\Nat}{O},
    \EqUpTo{\alpha}{\beta}{n}
    \to
    \IdType{\AppI{f}{\alpha}}{\AppI{f}{\beta}}{}
    \).
    \item \definiendum{uniformly continuous} if
    \(
    \exists\,n:\Nat,
    \forall\,\alpha,\beta:\ArrowType{\Nat}{O},
    \EqUpTo{\alpha}{\beta}{n}
    \to
    \IdType{\AppI{f}{\alpha}}{\AppI{f}{\beta}}{}
    \).
  \end{itemize}
  where we have used \(\EqUpTo{\alpha}{\beta}{n}\) to mean that \(\alpha\) and
  \(\beta\) share the same initial segment of length \(n\).
\end{definition}

\begin{remark}\label{rem:dialogue_continuous_implies_continuous}
  Due to their inductive nature, any path along a dialogue tree must reach a
  leaf node in a finite number of steps and hence query the oracle a finite
  number of times. This means that given a dialogue tree
  \(d:\Dialogue{\Nat}{O}{X}\), the function
  \(\AppI{\DialogueFSym}{d} : \ArrowType{(\ArrowType{\Nat}{O})}{X}\) must be
  continuous. Further assuming that the type \(O\) is finite, for example in the
  case of \(\Bool\), then the dialogue tree \(d\) can even be fully searched for
  the largest query over all paths, in which case \(\AppI{\DialogueFSym}{d}\)
  will also be uniformly continuous.
\end{remark}

\subsection{A Dialogue Tree Translation of \VSystemT}%
\label{sub:dialogue_translation}

As mentioned, dialogue trees form a monad and it is this structure
that allows us to define the dialogue tree translation of \VSystemT{}.

\begin{definition}[\AgdaLink{Definition-9}]\label{def:kleisli_ext}
  The \definiendum{Kleisli extension of dialogue trees} is defined inductively
  by
  \begin{gather*}
    \DeclareType%
    {\KleisliSym}%
    {\ArrowType{(\ArrowType{X}{\Dialogue{I}{O}{Y}})}{\ArrowType{\Dialogue{I}{O}{X}}{\Dialogue{I}{O}{Y}}}}\\
    \AppI{\Kleisli{f}}{\Leaf{x}} \IsDefinedToBe f(x) \\
    \AppI{\Kleisli{f}}{\Branch{\phi}{i}} \IsDefinedToBe \Branch{\UFun{x}{\AppI{\Kleisli{f}}{\AppI{\phi}{x}}}}{i}
  \end{gather*}
\end{definition}

The Kleisli extension of a function \(f\) will apply it at the leaves and graft
in the resulting trees, leaving intermediate nodes unchanged. With this we may
define the functorial action.

\begin{definition}[\AgdaLink{Definition-10}]\label{def:functor}
  The \definiendum{functorial action of dialogue trees} is defined by
  \begin{gather*}
    \DeclareType{\FunctorSym}{\ArrowType{(\ArrowType{X}{Y})}{\ArrowType{\Dialogue{I}{O}{X}}{\Dialogue{I}{O}{Y}}}}\\
    \NewDefinition{\Functor{f}}{\Kleisli{\LeafSym \circ f}}
  \end{gather*}
\end{definition}

To interpret the higher-order recursion in \VSystemT{} there are two possible
approaches. In one, we choose to interpret \VSystemT{} types as algebras over
the dialogue tree monad, as done in~\cite{Sterling:jfp:2021}. Alternatively, one
may interpret \VSystemT types as metatheoretic types (akin to the set model) and
use a generalised Kleisli extension to interpret higher-order recursion. The
former gives a compositional semantics which is simpler to define, but as we are
extending the formalisation of~\cite{mhe-effectful-forcing}, we choose the
latter. The following definition refers to the translation of \VSystemT{} types,
which can be found in \Cref{fig:dialogue-model}. One should see generalised
Kleisli extension as a pointwise version of Kleisli extension.

\begin{definition}[\AgdaLink{Definition-11}]\label{def:gen_kleisli_ext}
  The \definiendum{generalised Kleisli extension of dialogue trees} is defined
  by induction on \VSystemT{} types as follows
\begin{gather*}
  \DeclareType%
    {\GKleisli{\sigma}}%
    {\ArrowType{(\ArrowType{X}{\typestodial{\sigma}})}%
               {\ArrowType{\DialogueNN{X}}{\typestodial{\sigma}}}}\\
  \NewDefinition{\GKleisli{\TNat}}{\KleisliSym}\\
  \NewDefinition{\GKleisli{\TArrowType{\sigma_{1}}{\sigma_{2}}}}{%
    \UFun{f}{%
    \UFun{d}{%
    \UFun{s}{\AppII{\GKleisli{\sigma_{2}}}{\UFun{x}{\AppII{f}{x}{s}}}{d}}}}}
\end{gather*}
\end{definition}

\begin{figure}
  \small
\hspace*{-0.2in}
 \begin{tabular}{ll}
   \begin{tabular}{l}
     \(
     \begin{aligned}
       \DeclareType{\typestodial{-}&}{\ArrowType{\TTypes}{\UU}} \\
       \NewDefinition{\typestodial{\TNat}&}{\DialogueNN{\Nat}} \\
       \NewDefinition{\typestodial{\sigma_{0}\tto\sigma_{1}}&}{\ArrowType{\typestodial{\sigma_{0}}}{\typestodial{\sigma_{1}}}}
     \end{aligned}
     \)\\ \\
     \(
     \begin{aligned}
       \DeclareType{\ctxstodial{-}&}{\ArrowType{\TCtxs}{\UU}}\\
       \NewDefinition{\ctxstodial{\Gamma}&}{\ArrowType{(x:\sigma)\in\Gamma}{\typestodial{\sigma}}}
     \end{aligned}
     \)
   \end{tabular}
   &
\hspace*{-0.1in}
   \(
   \begin{aligned}
     \DeclareType{\termstodial{-}{-} &}{\ArrowType{\TTerms{\Gamma}{\sigma}}{\ArrowType{\ctxstodial{\Gamma}}{\typestodial{\sigma}}}}\\
     \NewDefinition{\termstodial{x}{\gamma}&}{\AppI{\gamma}{x}}\\
     \NewDefinition{\termstodial{\TZero}{\gamma}&}{\Leaf{\Zero}}\\
     \NewDefinition{\termstodial{\TSucc{t}}{\gamma}&}{\AppI{\Functor{\SuccSym}}{\termstodial{t}{\gamma}}}\\
     \NewDefinition{\termstodial{\TRec{\sigma}{t_1}{t_2}{t_3}}{\gamma}&}{\AppII{\GKleisli{\sigma}}{{\AppII{\RecSym}{\termstodial{t_1}{\gamma}\circ\LeafSym}{\termstodial{t_2}{\gamma}}}}{\termstodial{t_3}{\gamma}}}\\
     \NewDefinition{\termstodial{\TFun{x}{\sigma}{t}}{\gamma}&}{\Fun{x^\prime}{\typestodial{\sigma}}{\termstodial{t}{\gamma, (x \mapsto x^\prime)}}}\\
     \NewDefinition{\termstodial{\TAppI{t_1}{t_2}}{\gamma}&}{\AppI{\termstodial{t_1}{\gamma}}{\termstodial{t_2}{\gamma}}}
   \end{aligned}
   \)
 \end{tabular}
\caption{Dialogue interpretation of \VSystemT{}.}%
\label{fig:dialogue-model}
\end{figure}

For the dialogue interpretation, we will interpret the ground type \(\TNat\) as
the type \(\DialogueNN{\Nat}\) of dialogue trees over the natural numbers, and
the function type \(\tto\) is, as in the set model, interpreted by the
metatheoretic function type \(\mlto\).

\begin{definition}[\AgdaLink{Definition-12a}]\label{def:dialogue_interpretation}
  The \definiendum{dialogue interpretation} of types, contexts and terms of
  \VSystemT{} is defined in \Cref{fig:dialogue-model}. When defining the
  interpretation of terms we let \(\gamma\) range over~\(\ctxstodial{\Gamma}\).
  We fix the notation \(\gamma,(x \mapsto x^\prime)\) to extend an assignment of
  \(\ctxstodial{\Gamma}\) to an assignment
  of~\(\ctxstodial{\Gamma,(x:\sigma)}\).
\end{definition}

This translation differs from the one in~\cite{mhe-effectful-forcing} in two
ways. The first is that we are no longer using a combinator version of
\VSystemT{}. The inclusion of variables aids in internalizing metatheoretic
functions, which we will do frequently in future sections, hence it is desirable
even if it complicates the formalisation. Aside from this, we also do not extend
\VSystemT{} with an oracle, which as we will note later, turns out to be
unnecessary.

The final piece of the puzzle for effectful forcing is the existence of a
generic sequence in this new interpretation. To compute dialogue trees we require the
existence of a \emph{generic sequence}, that is, we need
\(\DeclareType{f}{\ArrowType{\DialogueNN{\Nat}}{\DialogueNN{\Nat}}}\) such that
for all \(\DeclareType{\alpha}{\ArrowType{\Nat}{\Nat}}\), the following commutes
\begin{center}
\begin{tikzcd}
  \DialogueNN{\Nat} \arrow[rr, "f"] \arrow[d, "\DialogueF{-}{\alpha}"']
  && \DialogueNN{\Nat} \arrow[d, "\DialogueF{-}{\alpha}"] \\
  \nat \arrow[rr, "\alpha"']
  && \nat
\end{tikzcd}
\end{center}

\begin{definition}[\AgdaLink{Definition-13}]\label{def:generic}
  The \definiendum{generic sequence} is defined as the following
  Kleisli extension:
  \begin{gather*}
    \DeclareType{\Generic}{\ArrowType{\DialogueNN{\Nat}}{\DialogueNN{\Nat}}}\\
    \NewDefinition{\Generic}{\Kleisli{{\AppI{\BranchSym}{\LeafSym}}}}
  \end{gather*}
\end{definition}

We can show this sequence to indeed satisfy the mentioned commuting diagram. % and
% though we call it \emph{the} generic sequence, there are many other functions
% satisfying this specification. Nonetheless, any such function will do and this
% seems the most natural choice.

\begin{definition}[\AgdaLink{Definition-14}]\label{def:dialogue_tree_operator}
  We define the \definiendum{dialogue tree operator} by
  \begin{gather*}
    \DeclareType%
    {\DialogueTree}%
    {\ArrowType{\TCTerms{\TArrowType{(\TArrowType{\TNat}{\TNat})}{\TNat}}}{\DialogueNN{\Nat}}}\\
    \NewDefinition{\AppI{\DialogueTree}{t}}{\AppI{\termstodial{t}{}}{\Generic}}
  \end{gather*}
\end{definition}

In its original formulation, a new oracle term
\(\ctyped{\Omega}{\TArrowType{\TNat}{\TNat}}\) was added to \VSystemT{}. Under
the dialogue translation, this term was interpreted by \(\Generic\) and given
\(\ctyped{t}{\TArrowType{(\TArrowType{\TNat}{\TNat})}{\TNat}}\) we could compute
dialogue tree for \(t\) by taking the dialogue translation of
\(\TAppI{t}{\Omega}\). This extension turns out to be unnecessary, using instead
the above definition of \(\DialogueTree\) and modifying the logical relation
used to prove correctness.

\begin{definition}[\AgdaLink{Definition-15}]
  Given a \VSystemT{} type \(\sigma\), a sequence
  \(\DeclareType{\alpha}{\ArrowType{\Nat}{\Nat}}\), and
  elements~\(x : \typestoset{\sigma}\) and~\(y : \typestodial{\sigma}\), we
  define the \definiendum{effectful forcing logical relation} by induction
  on~\(\sigma\):
  \begin{mathpar}
    \boxed{\Rdial{\alpha}{x}{y}{\sigma}}

    \inferrule{\IdType{n}{\DialogueF{d}{\alpha}}{}}{\Rdial{\alpha}{n}{d}{\TNat}}

    \inferrule%
    {\forall\,x : \typestoset{\sigma_1},\,%
     \forall\,y : \typestodial{\sigma_1},\,%
      \ArrowType{(\Rdial{\alpha}{x}{y}{\sigma_1})}{(\Rdial{\alpha}{\AppI{f}{x}}{\AppI{g}{y}}{\sigma_2})}}%
    {\Rdial{\alpha}{f}{g}{\sigma_1 \tto \sigma_2}}
  \end{mathpar}
\end{definition}

From this relation's fundamental lemma we can derive the
correctness result from~\cite{mhe-effectful-forcing}.

\begin{theorem}[\AgdaLink{Theorem-16} Correctness of \(\DialogueTree\)]\label{thm:correctness_dialogue_tree}
  For all sequences \(\DeclareType{\alpha}{\ArrowType{\Nat}{\Nat}}\) and closed
  terms \(\ctyped{t}{\TArrowType{(\TArrowType{\TNat}{\TNat})}{\TNat}}\) we have
  \( \AppI{\termstoset{t}{}}{\alpha} = \DialogueF{\AppI{\DialogueTree}{t}}{\alpha} \).
\end{theorem}

This result shows that any \VSystemT{}-definable functionals
\(F:\ArrowType{(\ArrowType{\Nat}{\Nat})}{\Nat}\) are dialogue
continuous and hence continuous. Furthermore, by encoding
\(\Bool\) in \VSystemT{}, we also see that any functionals
\(F:\ArrowType{(\ArrowType{\Nat}{\Bool})}{\Nat}\) definable in \VSystemT{}
must be uniformly continuous.

As observed in~\cite{Sterling:jfp:2021,Baillon+Mahboubi+Pedrot:csl:2022}, this
translation is not a model of \VSystemT{} as the dialogue interpretation of
\(\TRecSym\) does not satisfy the usual \VSystemT{} equations. For example, it
does \emph{not} validate the conversion
\(
  \TRec{\TNat}{f}{x}{\TSucc{n}} \equiv \TAppII{f}{n}{\TRec{\TNat}{f}{x}{n}}
\)
under the context
\(
  (f:\TArrowType{\TNat}{\TArrowType{\TNat}{\TNat}}),\,(x:\TNat),\,(n:\TNat)
\)
due to the existence of effectful terms.
As a result, convertible terms may be assigned completely different dialogue
trees, but this doesn't affect the correctness of our claims.

%%% Local Variables:
%%% mode: latexmk
%%% TeX-master: "../internal-effectful-forcing"
%%% End:

%% AT -- Church encodings of dialogue trees
%\input{sections/encodings.tex}

%% BRP -- The dialogue tree translation of \SystemT
\section{Internalising Dialogue Trees in \VSystemT}%
\label{sec:translation}

We now work on recreating the translation of \Cref{sub:dialogue_translation}
internally to \VSystemT{} using the Church encodings of dialogue trees. The main
difference from the inductive case is that due to the use of Church encodings,
we must now parameterise the translation by a \VSystemT{} type: that is, the
motive for the elimination principle of the Church encoded dialogue trees. For
the purpose of defining moduli of continuity it suffices to use the motive
\(A \IsDefinedToBe \TArrowType{(\TArrowType{\TNat}{\TNat})}{\TNat}\). With this
we could then define the dialogue tree following \Cref{sub:translation} and
compute a modulus of continuity according to \Cref{sec:computing-moduli}. The
insufficiency of considering a single motive reveals itself only when proving
the correctness of the Church encoded translation. For the correctness proof we
require multiple motives, for example \(A \IsDefinedToBe \TNat\) is used in the
\(\TRecSym\) case of \Cref{lem:main_lemma}.

\subsection{Church-Encoded Trees in \VSystemT}%
\label{sub:translation}

For the rest of this subsection we fix a motive \(A : \TTypes\) and give the
corresponding internal dialogue translation. We will also be slightly less
general with the internal definitions of dialogue trees as we will only need to
talk about the internalisation of \(\DialogueNN{\Nat}\).

%% Called ⌜η⌝ and ⌜\beta⌝ in the formalisation.
\begin{definition}[\AgdaLink{Definition-17a}]\label{def:internal_dialogue_trees}
  Given \(\sigma : \TTypes\), we define
  \definiendum{internal \(\sigma\)-dialogue trees} as:
  \begin{gather*}
    \DeclareType{\ChurchIntNN{\sigma}{A}}{\TTypes}\\
    \NewDefinition%
    {\ChurchIntNN{\sigma}{A}}%
    {\TArrowType%
      {(\TArrowType{\sigma}{A})}%
      {\TArrowType%
        {(\TArrowType{(\TArrowType{\TNat}{A})}{\TArrowType{\TNat}{A}})}%
        {A}
      }
    }
  \end{gather*}
  The constructors corresponding to \(\LeafSym\) and \(\BranchSym\) are
  the functions \(\LeafIntSym{A}\) and \(\BranchIntSym{A}\) defined as
  \[
    \begin{array}{l@{\hspace{0.2in}}l}
      \begin{array}{l}
      \DeclareType%
      {\LeafIntSym{A}}%
      {\TCTerms{
          \TArrowType{\sigma}{\ChurchIntNN{\sigma}{A}}
        }}\\
      \NewDefinition
      {\LeafIntSym{A}}%
      {\TUFun{z}{\TUFun{e}{\TUFun{b}{\TAppI{e}{z}}}}}\\
      \end{array} &
      \begin{array}{l}
      \DeclareType%
      {\BranchIntSym{A}}%
      {\TCTerms{
          \TArrowType{(\TArrowType{\TNat}{\ChurchIntNN{\sigma}{A}})}{\TArrowType{\TNat}{\ChurchIntNN{\sigma}{A}}}
        }}\\
      \NewDefinition%
      {\BranchIntSym{A}}%
      {\TUFun{\phi}{\TUFun{x}{\TUFun{e}{\TUFun{b}{\TUFun{y}{\TAppIII{\varphi}{y}{e}{b}}}}}}}
      \end{array}
    \end{array}
  \]
\end{definition}

\begin{definition}[\AgdaLink{Definition-18}]\label{def:internal_kleisli}
  The \definiendum{internal Kleisli extension} is the following closed
  \VSystemT{} term
  \begin{gather*}
  \DeclareType{\TKleisliSym{A}}{\TCTerms{\TArrowType{(\TArrowType{\TNat}{\ChurchIntNN{\TNat}{A}})}{\TArrowType{\ChurchIntNN{\TNat}{A}}{\ChurchIntNN{\TNat}{A}}}}}\\
    \NewDefinition{\TKleisliSym{A}}{%
      \TUFun{f     }{%
      \TUFun{d     }{%
      \TUFun{\eta' }{%
      \TUFun{\beta'}{
        \TAppII{d}{\TFun{x}{\iota}{\TAppIII{f}{x}{\eta'}{\beta'}}}{\beta'}%
    }}}}}
  \end{gather*}
\end{definition}

\begin{definition}[\AgdaLink{Definition-19}]\label{def:internal_functor}
  The \definiendum{internal functor action} is the following
  closed \VSystemT{} term
  \begin{gather*}
    \DeclareType{\TFunctorSym{A}}{\TCTerms{\TArrowType{(\TArrowType{\TNat}{\TNat})}{\TArrowType{\ChurchIntNN{\TNat}{A}}{\ChurchIntNN{\TNat}{A}}}}}\\
    \NewDefinition{\TFunctorSym{A}}{%
    \TUFun{f}{\TKleisli{\TUFun{x}{\LeafInt{\AppI{f}{x}}{A}}}{A}
    }}
  \end{gather*}
\end{definition}

For the following definition we use the translation of types
defined ahead in \Cref{fig:internal-translation}.

\begin{definition}[\AgdaLink{Definition-20}]\label{def:internal_gkleisli}
  The \definiendum{generalised internal Kleisli extension} is a family of closed
  \VSystemT{} terms indexed by \(\sigma : \TTypes\). These are defined
  by induction on the structure of \(\sigma\).
  \begin{gather*}
    \DeclareType{\TGKleisli{\sigma}{A}}%
    {\ArrowType%
      {(\sigma : \TTypes)}%
      {\TCTerms{\TArrowType{(\TArrowType{\TNat}{\typestoint{\sigma}{A}})}{\TArrowType{\ChurchIntNN{\TNat}{A}}{\typestoint{\sigma}{A}}}}}}\\
    \NewDefinition{\TGKleisli{\TNat}{A}}{\TKleisliSym{A}}\\
    \NewDefinition{\TGKleisli{\TArrowType{\sigma_{1}}{\sigma_{2}}}{A}}{%
      \TUFun{f}{%
        \TUFun{d}{%
          \TUFun{s}{%
            \AppII{\TGKleisli{\sigma_{2}}{A}}{\TUFun{x}{\TAppII{f}{x}{s}}}{d}
          }}}}
  \end{gather*}
\end{definition}

As made clear from the type signature, \(\TGKleisli{\sigma}{A}\) depends on the
\VSystemT{} type~\(\sigma\), which is only available in the metatheory. During
the internal translation we will always know what \(\sigma\) at which point
we get a fixed \VSystemT{} term. With this we are now able to define
the internal translation and the associated dialogue tree operators.

\begin{figure}
 \small
 \centering
 \begin{tabular}{m{0.5\textwidth}m{0.4\textwidth}}
   \(
   \begin{aligned}
     \DeclareType{\typestoint{-}{A}&}{\ArrowType{\TTypes}{\TTypes}} \\
     \NewDefinition{\typestoint{\TNat}{A}&}{\ChurchIntNN{\TNat}{A}} \\
     \NewDefinition{\typestoint{\sigma_{0}\tto\sigma_{1}}{A}&}{\TArrowType{\typestoint{\sigma_{0}}{A}}{\typestoint{\sigma_{1}}{A}}}
   \end{aligned}
   \) &
        \(
        \begin{aligned}
          \DeclareType{\ctxstoint{-}{A}&}{\ArrowType{\TCtxs}{\TCtxs}}\\
          \NewDefinition{\ctxstoint{\empctx}{A}&}{\empctx}\\
          \NewDefinition{\ctxstoint{\Gamma,(x:\sigma)}{A}&}{\ctxstoint{\Gamma}{A},(x:\typestoint{\sigma}{A})}
        \end{aligned}
        \)\\\\
   \multicolumn{2}{c}{
   \(
   \begin{aligned}
    \DeclareType{\termstoint{-}{A}&}{\ArrowType{\TTerms{\Gamma}{\sigma}}{\TTerms{\ctxstoint{\Gamma}{A}}{\typestoint{\sigma}{A}}}}\\
    \NewDefinition{\termstoint{x}{A}&}%
      {x}\\
    \NewDefinition{\termstoint{\TZero}{A}&}%
      {\LeafInt{\TZero}{A}}\\
    \NewDefinition{\termstoint{\TSucc{t}}{A}&}%
      {\TAppI{\TFunctor{\TSuccSym}{A}}{\termstoint{t}{A}}}\\
    \NewDefinition{\termstoint{\TRec{\sigma}{t_1}{t_2}{t_3}}{A}&}%
      {\TAppII{\TGKleisli{\sigma}{A}}{\TAppII{\TRecSym_{\sigma}}{\TUFun{x}{\TAppI{\termstoint{t_1}{A}}{\LeafInt{x}{A}}}}{\termstoint{t_2}{A}}}{\termstoint{t_3}{A}}}\\
    \NewDefinition{\termstoint{\TFun{x}{\sigma}{t}}{A}&}%
      {\TFun{x}{\typestoint{\sigma}{A}}{\termstoint{t}{A}}}\\
    \NewDefinition{\termstoint{\TAppI{t_1}{t_2}}{A}&}%
      {\TAppI{\termstoint{t_1}{A}}{\termstoint{t_2}{A}}}
   \end{aligned}
   \)
   }
 \end{tabular}
 \caption{Internal dialogue translation of \VSystemT{} with motive
   \(A : \TTypes\)}%
\label{fig:internal-translation}
\end{figure}

\begin{definition}[\AgdaLink{Definition-21a}]
  The \definiendum{internal dialogue translation} of \VSystemT{} is defined in
  \Cref{fig:internal-translation}.
\end{definition}

\begin{definition}[\AgdaLink{Definition-22}]\label{def:internal_generic}
  The \definiendum{internal generic sequence} is given by the term
 \begin{gather*}
   \DeclareType{\TGeneric{A}}{\TCTerms{\TArrowType{\ChurchIntNN{\TNat}{A}}{\ChurchIntNN{\TNat}{A}}}}\\
   \NewDefinition{\TGeneric{A}}{\TKleisli{\TAppI{\BranchIntSym{A}}{\LeafIntSym{A}}}{A}}
 \end{gather*}
\end{definition}

\begin{definition}[\AgdaLink{Definition-23}]\label{def:internal_dialogue_tree_operator}
  The \definiendum{internal dialogue tree operator} is the metatheoretic function
  \begin{gather*}
    \DeclareType{\TDialogueTreeSym{A}}{\ArrowType{\TCTerms{\TArrowType{(\TArrowType{\TNat}{\TNat})}{\TNat}}}{\TCTerms{\ChurchIntNN{\TNat}{A}}}}\\
    \NewDefinition{\TDialogueTree{t}{A}}{\TAppI{\termstoint{t}{A}}{\TGeneric{A}}}
  \end{gather*}
\end{definition}

As made explicit in the type signature, this operator lives in the metatheory but for any term
\(\ctyped{t}{\TArrowType{(\TArrowType{\TNat}{\TNat})}{\TNat}}\)
gives another term \(\ctyped{\TDialogueTree{t}{A}}{\ChurchIntNN{\TNat}{A}}\).
Of course, if we could fully internalise this dialogue operator then
we could further implement a moduli of continuity operator inside
\VSystemT, which as we discussed, is impossible.

\subsection{Avoiding Function Extensionality}%
\label{sub:avoiding_funext}

Without functional extensionality in the metatheory, it is not provable that the
metatheoretical equality and extensional equality of functions coincide. The
correctness results we are interested, e.g.\
\Cref{thm:correctness_internal_dialogue_tree}, are about equalities of natural
numbers, so intuitively they should not rely crucially on function
extensionality. It is not so simple however, for in the proofs of said
results we will quickly encounter cases where we must reason about
equality of higher-order functions. The solution is to use
hereditarily extensional equality~\cite{Troelstra:1973} instead of
the metatheoretic or general extensional equality.

\begin{definition}[\AgdaLink{Definition-24}]\label{defn:her-ext-equality}
  Given a \VSystemT{} type~\(\sigma\) and elements \(x,y:\typestoset{\sigma}\),
  we define the \definiendum{hereditarily extensionally equality} by induction
  on~\(\sigma\):
  \begin{mathpar}
    \boxed{\HEEquals{x}{y}{\sigma}}

    \inferrule{\IdType{n}{m}{\nat}}{\HEEquals{n}{m}{\TNat}}

    \inferrule%
    {\forall\,x, y : \typestoset{\sigma_1},\,\ArrowType{\HEEquals{x}{y}{\sigma_1}}{\HEEquals{\AppI{f}{x}}{\AppI{g}{y}}{\sigma_2}}}%
    {\HEEquals{f}{g}{\sigma_1 \tto \sigma_2}}
  \end{mathpar}
\end{definition}

When clear from context we omit the type annotation from the relation symbol. At
\(\TNat\), we can see that hereditarily extensional equality coincides with with
equality of natural numbers, and at \(\TNat \tto \TNat\) it coincides with
extensional equality of functions \(\ArrowType{\nat}{\nat}\). As we look at
higher types it diverges from extensional equality of functions and in general
we will not be able to prove it reflexive. Fortunately, all functions we need
are provably well-behaved.

\begin{lemma}[\AgdaLink{Lemma-25a}]\label{lem:hee_reflexive}
  For all \VSystemT{} types \(\sigma\), the binary relation
  \(\HEEquals{}{}{\sigma}\) is symmetric and transitive. Furthermore, if
  \(\sigma\) is of the shape
  \(\sigma \vcentcolon\vcentcolon= \TNat\ |\ \TArrowType{\TNat}{\sigma}\) then
  it is also reflexive.
\end{lemma}

\begin{lemma}[\AgdaLink{Lemma-26}]\label{lem:t-definable-reflexive}
  For all closed \(\ctyped{t}{\sigma}\), we have that
  \(\HEEquals{\termstoset{t}{}}{\termstoset{t}{}}{\sigma}\).
\end{lemma}

\subsection{Correctness of the Syntactic Translation}%
\label{sub:translation_correctness}

For the correctness of the internal dialogue translation we expect it to agree
with the more familiar inductive dialogue translation, meaning we must somehow
equate inductive dialogue trees with their Church-encodings. We can already
turn \VSystemT{} dialogue trees into metatheoretic Church-encoded trees
with the set semantics \(\termstoset{-}{}\). For the inductive trees,
it suffices to use the following encoding function, from which we define
the correctness relation.

%% Called church-encode
\begin{definition}[\AgdaLink{Definition-27}]\label{def:encode}
  Given a type \(A : \TTypes\) we define the following encode function
  by induction
  \begin{gather*}
    \DeclareType%
    {\EncodeExtSym{A}}%
    {\ArrowType%
      {\DialogueNN{\Nat}}%
      {\typestoset{\ChurchIntNN{\TNat}{A}}}
    }\\
    \NewDefinition%
    {\EncodeExt{\Leaf{z}}{A}}%
    {\AppI{\termstoset{\LeafIntSym{A}}{}}{z}}\\
    \NewDefinition%
    {\EncodeExt{\Branch{\phi}{x}}{A}}%
    {\AppII{\termstoset{\BranchIntSym{A}}{}}{\EncodeExtSym{A} \circ \phi}{x}}
  \end{gather*}
\end{definition}

\begin{definition}[\AgdaLink{Definition-28} Dialogue Correctness Logical Relation]
  \label{def:internal-logical-relation}
  Given \(\sigma : \TTypes\), \(x : \typestodial{\sigma}\) and a
  family of closed terms
  \(y : \ArrowType{(A : \TTypes)}{\TCTerms{\typestoint{\sigma}{A}}}\),
  we define the \definiendum{dialogue correctness logical relation} by induction
  on~\(\sigma\):
  \begin{mathpar}
    \boxed{\Rnorm{x}{y}{\sigma}}

    \inferrule{\forall\,A : \TTypes,\,\HEEquals{\EncodeExt{d}{A}}{t_A}{\TNat} }{\Rnorm{d}{t}{\TNat}}

    \inferrule%
    {\forall\,x : \typestodial{\sigma_1},\,%
     \forall\,y : \ArrowType{(A : \TTypes)}{\TCTerms{\typestoint{\sigma_2}{A}}},\,%\qquad\qquad\qquad\\\qquad\qquad%
      \ArrowType{\Rnorm{x}{y}{\sigma_1}}{\Rnorm{\AppI{f}{x}}{\Fam{\TAppI{g_A}{y_A}}{A:\TTypes}}{\sigma_2}}}%
    {\Rnorm{f}{g}{\sigma_1 \tto \sigma_2}}
  \end{mathpar}
  We extend this relation to act pointwise on contexts \(\Gamma\) and
  denote this extension \(\Rnorm{\gamma_{1}}{\gamma_{2}}{\Gamma}\).\todo{maybe should define this properly}
\end{definition}

For a number of the proofs it will be important that normalisation preserves this
relation. For our purposes it suffices to compute in the Set model,
once again avoiding conversion rules.

\begin{lemma}[\AgdaLink{Lemma-29}]\label{lem:Rnorm_respects_conversion}
  Fix two family of terms
  \(t,s : \ArrowType{(A:\TTypes)}{\TCTerms{\typestoint{\sigma}{A}}}\)
  and~\(x : \typestodial{\sigma}{}\). If for all \(A : \TTypes\) we have
  \(\HEEquals{\termstoset{t_{A}}{}}{\termstoset{s_{A}}{}}{\typestoint{\sigma}{A}}\)
  then \(\Rnorm{x}{t}{\sigma}\) implies \(\Rnorm{x}{s}{\sigma}\).
\end{lemma}
\begin{proof}
  We proceed by induction on the \VSystemT{} type \(\sigma\). If \(\sigma\) is
  of the form \(\TNat\) then the result follows by symmetry and transitivity of
  \(\HEEqualsSymbol\). If \(\sigma\) is a function type then we apply the inductive
  hypothesis pointwise.
\end{proof}

With this we can now prove some lemmas about the monadic operations of internal
dialogue trees. These will be used in the proof of
the fundamental lemma, that is in~\Cref{lem:main_lemma}.

\begin{lemma}[\AgdaLink{Lemma-30}]\label{lem:church_encode_kleisli}
  Given a type \(A : \TTypes\), a dialogue tree \(\DeclareType{d}{\typestodial{\TNat}}\),
  and functions
  \(\DeclareType{f_1}{\ArrowType{\Nat}{\typestodial{\TNat}}}\) and
  \(\DeclareType{f_2}{\ArrowType{\Nat}{\typestoset{\ChurchIntNN{\TNat}{A}}}}\),
  if for all \(\DeclareType{n}{\Nat}\) we have
  \(
    \HEEquals%
    {\EncodeExt{\AppI{f_1}{n}}{A}}%
    {\AppI{f_2}{n}}%
    {\ChurchIntNN{\TNat}{A}}
  \)
  then we have in addition
  \(
    \HEEquals%
    {\EncodeExt{\AppI{\Kleisli{f_1}}{d}}{A}}%
    {\AppII{\termstoset{\TKleisliSym{A}}{}}{f_2}{\EncodeExt{d}{A}}}%
    {\ChurchIntNN{\TNat}{A}}
  \).
\end{lemma}
\begin{proof}
  We proceed by induction on the dialogue tree \(d\). In the case that \(d\) is
  of the form \(\Leaf{n}\), then the left-hand-side computes to
  \(\EncodeExt{\AppI{f_1}{n}}{A}\) and the right side computes to
  \(\termstoset{\TAppI{f_2}{\TNumeral{n}}}{}\), but these are equal by our
  assumption. In the case that \(d\) is of the form \(\Branch{\phi}{i}\), then
  we can apply the inductive hypothesis to each subtree \(\AppI{\phi}{n}\).
\end{proof}

\begin{corollary}[\AgdaLink{Corollary-31}]\label{lem:church_encode_natural}
  Given a type \(A : \TTypes\), a dialogue tree \(d : \DialogueNN{\Nat}\) and
  functions \(\DeclareType{g_1,g_2}{\ArrowType{\Nat}{\Nat}}\), if
  \(\HEEquals{g_1}{g_2}{}\) then
  \(
    \HEEquals%
    {\EncodeExt{\AppI{\Functor{g_2}}{d}}{A}}%
    {\AppII{\termstoset{\TFunctorSym{A}}{}}{g_1}{\EncodeExt{d}}{A}}%
    {\ChurchIntNN{\TNat}{A}}
  \).
\end{corollary}
\begin{proof}
  Instantiate \Cref{lem:church_encode_kleisli} with \(\LeafSym \circ g_{1}\)
  and \(\termstoset{\LeafIntSym{A}}{} \circ g_{2}\).
\end{proof}

\begin{lemma}[\AgdaLink{Lemma-32} Kleisli Lemma]\label{lem:kleisli_lemma}
  Fix a function
  \(
  \DeclareType{f}{\ArrowType{\Nat}{\typestodial{\sigma}}}
  \),
  a dialogue tree
  \(
  \DeclareType{n}{\DialogueNN{\Nat}}
  \),
  and families of terms
  \(
  \DeclareType%
  {g}%
  {\ArrowType{(A:\TTypes)}{\TCTerms{\typestoint{\sigma}{A}}}}
  \)
  and
  \(
  \DeclareType%
  {m}%
  {\ArrowType{(A:\TTypes)}{\TCTerms{\ChurchIntNN{\TNat}{A}}}}
  \).
  If
  \(
  \Rnorm{n}{m}{\TNat}
  \)
  and
  \(
  \forall i:\Nat, \Rnorm{\AppI{f}{i}}{\TAppI{g}{\TNumeral{i}}}{\sigma}
  \)\todo{applying g to i like this might be confusing}
  then
  \(
    \Rnorm%
    {\AppII{\GKleisli{\sigma}}{f}{n}}%
    {\AppII{\TGKleisli{\sigma}{A}}{g}{m}}%
    {\sigma}
  \).
\end{lemma}
\begin{proof}
  We proceed by induction on the \VSystemT{} type \(\sigma\). For the base case
  we must show for all \(A:\TTypes\) that
  \(
    \HEEquals%
    {\termstoset{\TAppI{\TKleisli{g_{A}}{A}}{m_{A}}}{}}%
    {\EncodeExt{\AppI{\Kleisli{f}}{n}}{A}}%
    {}
  \).
  By the assumption that \(\Rnorm{n}{m}{\TNat}\), the left hand side equals
  \(
  \AppII{\termstoset{\TKleisliSym{A}}{}}{\termstoset{g_{A}}{}}{\EncodeExt{n}{A}}
  \)
  and by \Cref{lem:church_encode_kleisli} this equals the right hand side. For
  the recursive step the inductive hypothesis suffices since generalised Kleisli
  extension at \(\TArrowType{\sigma_1}{\sigma_2}\) is defined in terms of
  \(\sigma_{2}\).
\end{proof}

\begin{lemma}[\AgdaLink{Lemma-33} Fundamental Lemma]\label{lem:main_lemma}
  Given \(\typed{\Gamma}{t}{\sigma}\), an assignment
  \(\gamma_1 : \typestodial{\Gamma}\) and a substitution \(\gamma_2\) from
  \(\empctx\) to \(\typestoint{\Gamma}{}\), if
  \( \Rnorm{\gamma_1}{\gamma_2}{\Gamma} \) then
  \( \Rnorm{\termstodial{t}{\gamma_1}}{\TSub{\termstoint{t}{}}{\gamma_2}}{\sigma} \).
\end{lemma}
\begin{proof}
  We proceed by induction on~\(t\).
  We will often implicitly use \Cref{lem:Rnorm_respects_conversion} whenever we
  must show a term satisfies the logical relation after some normalisation
  occurs.
  The variable case follows from the assumption that
  \(\Rnorm{\gamma_1}{\gamma_2}{\Gamma}\).
  The application case follows by induction on both subterms and the lambda
  abstraction case follows by expanding the substitution and applying the
  inductive hypothesis to the lambda body.
  The \(\TZero\) case amounts to showing that
  \(
  \HEEquals{\termstoset{\LeafInt{\TZero}{A}}{}}{\EncodeExt{\Leaf{\Zero}}{A}}{}
  \)
  which is seen to hold by expanding both definitions.
  The \(\TSucc{t_1}\) case follows by applying the inductive hypothesis to
  \(t_1\) and using \Cref{lem:church_encode_natural} to commute the
  \(\EncodeExtSym{A}\) to the outside.
  Finally, the \(\TRec{\sigma}{t_1}{t_2}{t_3}\) case follows from
  \Cref{lem:kleisli_lemma} which has two assumptions:
  for the first, we show for all \(n\) by induction that
  \(
    \Rnorm%
    {\Rec{\termstodial{t_1}{\gamma_1}}{\termstodial{t_2}{\gamma_1}}{n}}%
    {\TSub%
      {\TRec{}{\termstoint{t_1}{}\circ\LeafIntSym{A}}{\termstoint{t_2}{}}{\TNumeral{n}}}%
      {\gamma_2}}%
    {\sigma}
  \)
  using the inductive hypothesis from the subterms \(t_1\) and \(t_2\).
  The second assumption holds by the inductive hypothesis from the subterm
  \(t_{3}\).
\end{proof}

\begin{lemma}[\AgdaLink{Lemma-34} Dialogue tree agreement]\label{lem:dialogue_agreement}
  Given \(A : \TTypes\) and
  \(
  \ctyped{t}{\TArrowType{(\TArrowType{\TNat}{\TNat})}{\TNat}}
  \),
  we have
  \(
    \HEEquals%
    {\termstoset{\TDialogueTree{t}{A}}{}}%
    {\EncodeExt{\AppI{\DialogueTree}{t}}{A}}%
    {\ChurchIntNN{\TNat}{A}}
  \).
\end{lemma}
\begin{proof}
  It suffices to show
  \(
  \Rnorm%
  {\AppI{\termstodial{t}{}}{\Generic}}%
  {\Fam{\TAppI{\termstoint{t}{A}}{\TGeneric{A}}}{A:\TTypes}}%
  {\TNat}
  \).
  By \Cref{lem:main_lemma} we know that
  \(\Rnorm{\termstodial{t}{}}{\termstoint{t}{}}{\TArrowType{(\TArrowType{\TNat}{\TNat})}{\TNat}}\),
  leaving us to show \(\Rnorm{\Generic}{\TGeneric{A}}{\TArrowType{\TNat}{\TNat}}\).
  This follows from \Cref{lem:kleisli_lemma} as both the generic
  sequences are defined with kleisli extension.
\end{proof}

%% Called dialogue⋆ (in Internal.Correctness) in the formalisation.
%% TODO: Check if this is correct.
\begin{definition}[\AgdaLink{Definition-35}]\label{defn:dialogue_internal}
  We define the \definiendum{internal dialogue operator} as the closed term
  \begin{gather*}
    \DeclareType%
    {\TDialogueFSym}%
    {\TCTerms{\TArrowType%
      {\ChurchIntNN{\TNat}{\TArrowType{(\TArrowType{\TNat}{\TNat})}{\TNat}}}%
      {\TArrowType{(\TArrowType{\TNat}{\TNat})}{\TNat}}
    }}\\
    \NewDefinition{\TDialogueFSym}{%
      \TUFun{d}{\TAppII{d}{\TUFun{z}{\TUFun{\_}{z}}}{\TUFun{\phi}{\TUFun{x}{\TUFun{\alpha}{\TAppII{\phi}{\TAppI{\alpha}{x}}{\alpha}}}}}}}%
  \end{gather*}
\end{definition}

%% This is called dialogues-agreement (in MFPSAndVariations.Church).
\begin{lemma}[\AgdaLink{Lemma-36}]\label{lem:dialogues-agreement}
  Given a dialogue tree \(d : \DialogueNN{\Nat}\) and sequence
  \(\alpha : \ArrowType{\Nat}{\Nat}\) we have the
  following equality of natural numbers
  \(
    \IdType{\DialogueF{d}{\alpha}}{\AppII{\termstoset{\TDialogueFSym}{}}{\EncodeExt{d}{\TArrowType{(\TArrowType{\TNat}{\TNat})}{\TNat}}}{\alpha}}{}
  \).
\end{lemma}

\todo[inline]{the following proof was changed to use internal versions only,
  make sure to change it in the formalisation}

\begin{theorem}[\AgdaLink{Theorem-37} Correctness of \(\TDialogueTreeSym{A}\)]\label{thm:correctness_internal_dialogue_tree}
  Given
  \(\ctyped{t}{\TArrowType{(\TArrowType{\TNat}{\TNat})}{\TNat}}\) and a sequence
  \(\DeclareType{\alpha}{\ArrowType{\Nat}{\Nat}}\), we have
  \(
    \IdType%
    {\AppI{\termstoset{t}{}}{\alpha}}%
    {\AppI{\termstoset{\TAppI{\TDialogueFSym}{\TDialogueTree{t}{\TArrowType{(\TArrowType{\TNat}{\TNat})}{\TNat}}}}{}}{\alpha}}%
    {}
  \).\todo{i changed this to use \(\TDialogueFSym\) instead of the external one}
\end{theorem}
\begin{proof}
  By \Cref{lem:dialogue_agreement} we know that
  \(
    \HEEquals%
    {\termstoset{\TDialogueTree{t}{A}}{}}%
    {\EncodeExt{\AppI{\DialogueTree}{t}}{A}}%
    {}
  \).
  and by \Cref{lem:t-definable-reflexive} we also have
  \(
    \HEEquals%
    {\termstoset{\TDialogueFSym}{}}%
    {\termstoset{\TDialogueFSym}{}}%
    {}
  \).
  Composing these we get
  \(
    \IdType%
    {\termstoset{\TAppI{\TDialogueFSym}{\TDialogueTree{t}{A}}}{}}%
    {\AppI{\termstoset{\TDialogueFSym}{}}{\EncodeExt{\AppI{\DialogueTree}{t}}{\TArrowType{(\TArrowType{\TNat}{\TNat})}{\TNat}}}}%
    {}
  \),
  noting that we changed from \(\HEEqualsSymbol\) to the metatheoretic equality
  as they coincide for natural numbers.
  Using \Cref{lem:dialogues-agreement} we can replace the right hand side
  by
  \(
    \DialogueF{\AppI{\DialogueTree}{t}}{\alpha}
  \)
  and chaining this with \Cref{thm:correctness_dialogue_tree} we get
  \(\AppI{\termstoset{t}{}}{\alpha}\) as needed.
\end{proof}

%%% Local Variables:
%%% mode: latexmk
%%% TeX-master: "../internal-effectful-forcing"
%%% End:

%% AT
\section{Computing Moduli of Continuity Internally}%
\label{sec:computing-moduli}

As we have previously explained, the fundamental idea of a dialogue tree is to
encode information on how a \VSystemT{} term of type
$(\TNat \tto \TNat) \tto \TNat$ interacts with a given argument while computing
a result. In \Cref{sec:translation}, we presented the \VSystemT{} encodings of
such dialogue trees. We now proceed to define two \VSystemT{} operators that
compute \emph{moduli of continuity} using the information contained in these
internal trees:
\begin{enumerate}
  \item One, presented in this section, that takes an $\Nat$-branching dialogue
    tree and computes the modulus of continuity of the function it encodes
    \emph{at a given point} of the Baire space.
  \item Another one, presented in \Cref{sec:uniform}, that takes a
    binary-branching dialogue tree and computes the \emph{modulus of uniform
    continuity} of the function that it encodes.
\end{enumerate}

We assume we have a \(\MaxSym\) function on natural numbers as well as a
corresponding \VSystemT{} term
\(\TMaxSym : \TCTerms{\TArrowType{\TNat}{\TArrowType{\TNat}{\TNat}}}\), the
details of which can be found in our formalisation.

The main content of our modulus operator is given by a function that we call
$\MaxQuestionSym$, which computes the maximum question occurring on a given path
of a dialogue tree. We continue to follow our convention of implementing
constructions involving dialogue trees in two forms: on external inductive type
encodings, and on internal Church encodings. This is not only for the sake of
clarity but also to enable the precise formulation of the lemmas that we need
for our main result. Accordingly, we define the following two functions:
\begin{enumerate}
  \item $\MaxQuestionSym$ on external inductive type encodings.
  \item $\TMaxQuestionSym$ on internal Church encodings, implemented in
    \VSystemT{}.
\end{enumerate}

\begin{definition}[\AgdaLink{Definition-38}]\label{defn:max-question}
  We define the \definiendum{max question along a path} by induction on
  dialogue trees:
  \begin{gather*}
    \DeclareType{\MaxQuestionSym}{%
      \ArrowType{\DialogueNN{\Nat}}{\ArrowType{(\ArrowType{\Nat}{\Nat})}{\Nat}}%
    }\\
    \NewDefinition{\MaxQuestionSym(\Leaf{n})(\alpha)}{0}\\
    \NewDefinition{\MaxQuestionSym(\Branch{\phi}{n})(\alpha)}{%
      \Max{n}{\MaxQuestionSym(\phi(\alpha(n)))(\alpha)}
    }
  \end{gather*}
\end{definition}

\begin{definition}[\AgdaLink{Definition-39}]\label{defn:max-question-int}
  We define the \definiendum{internal max question along a path} function
  as
  {
   \begin{align*}
    \DeclareType{%
      &\TMaxQuestionSym
    }{%
      \TCTerms{\TArrowType{\ChurchIntNN{\TNat}{\TNat}}{\TArrowType{(\TArrowType{\TNat}{\TNat})}{\TNat}}}
    }\\
    \NewDefinition{&\TMaxQuestionSym}{%
      \TUFun{d}{%
        \TUFun{\alpha}{%
          \TAppII{d}{\TUFun{\_}{\TZero}}{\TUFun{g}{\TUFun{x}{\TMax{x}{\TAppI{g}{\TAppI{\alpha}{x}}}}}}
        }
      }
    }
   \end{align*}}
\end{definition}

\begin{lemma}[\AgdaLink{Lemma-40}]\label{lem:max-question-agreement-ext}
  For all dialogue trees \(d:\DialogueNN{\Nat}\) and sequences
  \(\alpha : \Nat \mlto \Nat\), we have that
  \(\AppII{\MaxQuestionSym}{d}{\alpha} = \AppII{\termstoset{\TMaxQuestionSym}{}}{\EncodeExt{d}{\TNat}}{\alpha}\).
\end{lemma}

We are now ready to define our modulus operators.

\begin{definition}[\AgdaLink{Definition-41a}]\label{defn:modulus}
  We define the \definiendum{external and internal modulus operators} as
  \[
    \begin{array}{l@{\hspace{0.2in}}l}
      \begin{array}{l}
      \DeclareType{\ModulusSym}{\ArrowType{\DialogueNN{\Nat}}{\ArrowType{(\ArrowType{\Nat}{\Nat})}{\Nat}}}\\
      \NewDefinition{\Modulus{d}{\alpha}}{\Succ{\MaxQuestion{d}{\alpha}}}\\
      \end{array} &
      \begin{array}{l}
      \DeclareType{\ModulusIntSym}{\TCTerms{\TArrowType{\ChurchIntNN{\TNat}{\TNat}}{\TArrowType{(\TArrowType{\TNat}{\TNat})}{\TNat}}}}\\
      \NewDefinition{\ModulusIntSym}{\TUFun{d}{\TUFun{\alpha}{\TSucc{\TMaxQuestion{d}{\alpha}}}}}
      \end{array}
    \end{array}
  \]
\end{definition}

Before we proceed to prove the correctness of the internal modulus of continuity
operator, we first formally define the notion of modulus of continuity.

\begin{definition}[\AgdaLink{Definition-42} Modulus of continuity]
  Let $f : (\Nat \mlto \Nat) \mlto \Nat$ be a function on the Baire space and
  let $\alpha : \ArrowType{\Nat}{\Nat}$ be a point. A natural number $m : \Nat$
  is a \definiendum{modulus of continuity} for $f$ at \(\alpha\) if
  the following holds:
  \(\forall \beta : \ArrowType{\Nat}{\Nat}.\ \EqUpTo{\alpha}{\beta}{m} \mlto
    f(\alpha) = f(\beta). \)
\end{definition}

We mentioned in \Cref{rem:dialogue_continuous_implies_continuous} that the
computation encoded by any dialogue tree is continuous. The function
$\ModulusSym$, which we defined above, can be seen as the \emph{computational
content} witnessing this fact.

\begin{lemma}\label{lem:mod-dialogue}
  Given any dialogue tree $d : \DialogueNN{\Nat}$ and any sequence
  $\alpha : \Nat \mlto \Nat$, $\Modulus{d}{\alpha}$ is a modulus of
  continuity of the function
  $\AppI{\DialogueFSym}{d} : (\Nat \to \Nat) \mlto \Nat$ at point~$\alpha$.
\end{lemma}

In order to prove the correctness of \(\ModulusIntSym\), we also need to know
that the internal modulus operator agrees with the external modulus operator. In
preparation for this, we prove the following lemma connecting the internal and
external processes of extracting moduli of continuity from the respective
encodings.

\begin{lemma}[\AgdaLink{Lemma-44}]\label{lem:main}
  For every term \(t : \TCTerms{\TArrowType{(\TArrowType{\TNat}{\TNat})}{\TNat}}\)
  and every point \(\alpha : \ArrowType{\Nat}{\Nat}\) of the Baire space, we have
  \(
  \AppI{\termstoset{\TAppI{\TMaxQuestionSym}{\TDialogueTree{t}{\TNat}}}{}}{\alpha}
  = \MaxQuestion{\AppI{\DialogueTree}{t}}{\alpha}
  \).
\end{lemma}
\begin{proof}
  Fix a term \(t : \TCTerms{\TArrowType{(\TArrowType{\TNat}{\TNat})}{\TNat}}\)
  and sequence \(\alpha : \ArrowType{\Nat}{\Nat}\). We have the following
  \begin{gather*}
  \AppI{\termstoset{\TAppI{\TMaxQuestionSym}{\TDialogueTree{t}{\TNat}}}{}}{\alpha} \\
  \quad=\AppII{\termstoset{\TMaxQuestionSym}{}}{\termstoset{\TDialogueTree{t}{\TNat}}{}}{\alpha}
  \tag{by \Cref{def:set-model}}\\
  \quad=\AppII{\termstoset{\TMaxQuestionSym}{}}{\EncodeExt{\AppI{\DialogueTree}{t}}{\TNat}}{\alpha}
  \tag{by \Cref{lem:dialogue_agreement}} \\
  \quad=\MaxQuestion{\AppI{\DialogueTree}{t}}{\alpha}
  \tag{by \Cref{lem:max-question-agreement-ext}}
  \end{gather*}
  which completes the proof.
\end{proof}

The key step in the above proof is the use of \Cref{lem:dialogue_agreement},
which relies on the logical relation from \Cref{def:internal-logical-relation}
and the hereditarily extensional equality relation from
\Cref{defn:her-ext-equality}.
With \Cref{lem:main} established, we may now proceed to prove our main result
for functions on the Baire space: $\ModulusIntSym$ computes moduli of continuity
for \VSystemT{}-definable functions on the Baire space.

\begin{theorem}[\AgdaLink{Theorem-45} Correctness of $\ModulusIntSym$]
  Let $t : (\TNat \tto \TNat) \tto \TNat$ be a \VSystemT{} function on the Baire
  space. The result $\termstoset{\TAppI{\ModulusIntSym}{t}}{}$ is a function
  giving a modulus of continuity for
  $\termstoset{t}{} : (\Nat \mlto \Nat) \mlto \Nat$ at each point of the Baire
  space.
\end{theorem}
\begin{proof}
  We know by \Cref{thm:correctness_dialogue_tree} that $t$ can be encoded by the
  inductive dialogue tree \(\AppI{\DialogueTree}{t}\), i.e.\ we have
  \(\AppI{\termstoset{t}{}}{\alpha} = \DialogueF{\AppI{\DialogueTree}{t}}{\alpha}\) for
  every $\alpha : \Nat \mlto \Nat$. We therefore know, by \Cref{lem:mod-dialogue},
  that $\ModulusSym$ gives moduli of continuity for $\termstoset{t}{}$, which is
  to say, we just have to show
  \(
    \AppI{\termstoset{\TAppI{\ModulusIntSym}{t}}{}}{\alpha}
    =
    \Modulus{\AppI{\DialogueTree}{t}}{\alpha},
  \)
  for all $\alpha : \Nat \mlto \Nat$, and this follows from \Cref{lem:main}.
\end{proof}

%%% Local Variables:
%%% mode: latexmk
%%% TeX-master: "../internal-effectful-forcing"
%%% End:

%% Define the operators:
%% --> `max`,
%% --> `max-question`, and
%% --> modulus.
%% Present the three versions of these operators for external inductive type
%% encodings, external Church encodings, and internal Church encodings.
%% Present the proofs of correctness of these.

%% AT
\section{Computing Moduli of Uniform Continuity Internally}%
\label{sec:uniform}

We now extend our work from the previous section as to define an operator
computing moduli of \emph{uniform continuity} of functions on the Cantor space.
Our development here closely follows that of \Cref{sec:computing-moduli}: we
implement an analogue of the $\MaxQuestionSym$ function from
\Cref{defn:max-question}, which we call $\MaxBoolQuestionSym$. This computes the
maximum question in a $\Bool$-branching dialogue tree. Unlike the
$\MaxQuestionSym$ operator, however, it performs this computation over the
entire tree rather than a specific path, exploiting the finiteness of the
branching in the context of the Cantor space.

Once again, we start by defining two operators:
\begin{enumerate}
  \item $\MaxBoolQuestionSym$ for external inductive type encodings.
  \item $\MaxBoolQuestionIntSym$ for internal Church encodings.
\end{enumerate}

\VSystemT{} does not include a type of Booleans, so we cannot directly talk
about the points of the Cantor space in it i.e.\ functions of type \(\alpha : \Nat
\mlto \Bool\). To avoid extending \VSystemT with another ground type, we work
with the embedding of the Cantor space into the Baire space, which we define
below.

\begin{definition}[\AgdaLink{Definition-46}]
  The embedding of the Cantor space into the Baire space is given by the
  function
  \(\EmbedCBSym : (\Nat \mlto \Bool) \to (\Nat \mlto \Nat)\), defined as
  \(\AppI{\EmbedCB{\alpha}}{i} = \Embed{\AppI{\alpha}{i}}\),
  where
  $\EmbedSym$ denotes the function mapping
  $\mathsf{false}$ to $\Zero$ and $\mathsf{true}$ to $\One$.
\end{definition}

The fact that we are working with such an embedding of the Cantor space implies
that we have to work with encodings of $\Bool$-branching dialogue trees into
$\Nat$-branching ones. Accordingly, we define a pruning function that converts a
$\Bool$-branching tree back into a $\Nat$-branching one.

\begin{definition}[\AgdaLink{Definition-47}]\label{def:pruning}
  We define a \definiendum{pruning operation} by induction on
  dialogue trees by
  \begin{gather*}
    \DeclareType{\PruneSym}{\DialogueNN{\Nat} \mlto \Dialogue{\Nat}{\Bool}{\Nat}}\\
    \NewDefinition{\Prune{\Leaf{n}}}{\Leaf{n}}\\
    \NewDefinition{\Prune{\Branch{\phi}{n}}}{%
      \Branch{\PruneSym \circ \phi \circ \EmbedSym}{n}
    }
  \end{gather*}
  where the $\EmbedSym$ function maps $\mathsf{false}$ to $\Zero$ and
  $\mathsf{true}$ to $\One$.
\end{definition}

We now proceed to define the binary versions of the $\MaxQuestionSym$ functions.

\begin{definition}[\AgdaLink{Definition-48}]\label{defn:max-bool-question}
  We define the function $\MaxBoolQuestionSym$ as follows:
  \begin{align*}
    \DeclareType{& \MaxBoolQuestionSym}{\ArrowType{\Dialogue{\Nat}{\Bool}{\Nat}}{\Nat}}\\
    \NewDefinition{& \MaxBoolQuestion{\Leaf{n}}}{\Zero}\\
    \NewDefinition{& \MaxBoolQuestion{\Branch{\phi}{n}}}{%
      \Max{n}{\Max{\MaxBoolQuestion{\AppI{\phi}{\Zero}}}{\MaxBoolQuestion{\AppI{\phi}{\One}}}}
    }
  \end{align*}
\end{definition}

\begin{definition}[\AgdaLink{Definition-49}]\label{def:internal_max_bool_question}
  We define the internal Church encoding version of $\MaxBoolQuestionSym$
  as
  \begin{align*}
    \DeclareType{& \MaxBoolQuestionIntSym}{%
      \TCTerms{\TArrowType{\ChurchIntNN{\TNat}{\TNat}}{\TNat}}
    }\\
    \NewDefinition{& \MaxBoolQuestionIntSym}{%
      \TUFun{d}{\TAppII{d}{\TUFun{\_}{\TZero}}{\TUFun{g}{\TUFun{x}{\TMax{x}{\TMax{\TAppI{g}{\TNumeral{\Zero}}}{\TAppI{g}{\TNumeral{\One}}}}}}}}
    }
  \end{align*}
\end{definition}

\begin{lemma}[\AgdaLink{Lemma-50}]\label{lem:max-bool-question-agreement-ext}
  Given a tree $d : \DialogueNN{\Nat}$, we have
  \(
    \IdTy{\MaxBoolQuestion{\Prune{d}}}{\AppI{\termstoset{\MaxBoolQuestionIntSym}{}}{\EncodeExt{d}{\TNat}}}
  \).
\end{lemma}

\begin{definition}[\AgdaLink{Definition-51a}]\label{def:uniform_modulus}
  We define the \definiendum{external and internal uniform modulus operators} as:
  \[
    \begin{array}{l@{\hspace{0.2in}}l}
      \begin{array}{l}
        \DeclareType{\ModulusUniSym}{\Dialogue{\Nat}{\Bool}{\Nat} \mlto \Nat}\\
        \NewDefinition{\ModulusUni{d}}{1 + \MaxBoolQuestion{d}}\\
      \end{array} &
      \begin{array}{l}
        \DeclareType{\ModulusUniIntSym}{\TCTerms{\TArrowType{\ChurchIntNN{\TNat}{\TNat}}{\TNat}}}\\
        \NewDefinition{\ModulusUniIntSym}{\TUFun{d}{\TSucc{\MaxBoolQuestionInt{d}}}}
      \end{array}
    \end{array}
  \]
\end{definition}

\begin{definition}[\AgdaLink{Definition-52}]
  Given a function on the Cantor space $f : (\Nat \mlto \Bool) \mlto \Nat$, a
  natural number $m : \Nat$ is said to be a \definiendum{modulus of uniform
  continuity} for~$f$ if the following holds:
  \(
    \forall \alpha, \beta : \Nat \mlto \Bool.\ \EqUpTo{\alpha}{\beta}{m} \to \AppI{f}{\alpha} = \AppI{f}{\beta}\text{.}
  \)
\end{definition}

From \Cref{rem:dialogue_continuous_implies_continuous}, we also know that the
computation encoded by any $\Bool$-branching dialogue tree is \emph{uniformly}
continuous. The function $\ModulusUniSym$ above can be seen as the computational
content of this fact.

\begin{lemma}[\AgdaLink{Lemma-53}]\label{lem:mod-uni-dialogue}
  Given any dialogue tree $d : \Dialogue{\Nat}{\Nat}{\Nat}$, the result
  $\ModulusUni{\Prune{d}}$ is a modulus of uniform continuity of the function
  $\AppI{\DialogueFSym}{\Prune{d}} : (\Nat \to \Bool) \mlto \Nat$.
\end{lemma}

Towards proving the correctness of the internal modulus of uniform continuity
operator in \Cref{thm:correctness-of-modulus-of-uniform-continuity}, we now
prove the following lemma, connecting the internal and external processes of
extracting moduli of uniform continuity from the respective encodings of
dialogue trees.

\begin{lemma}[\AgdaLink{Lemma-54}]\label{lem:main-uni}
  For all terms
  \(t : \TCTerms{\TArrowType{(\TArrowType{\TNat}{\TNat})}{\TNat}}\), the
  external and internal uniform max questions agree, i.e.\
  \(\IdTy{\termstoset{\MaxBoolQuestionInt{\TDialogueTree{t}{\TNat}}}{}}{\MaxBoolQuestion{\Prune{\AppI{\DialogueTree}{t}}}}\).
\end{lemma}
\begin{proof}
  Let \(t : \TCTerms{\TArrowType{(\TArrowType{\TNat}{\TNat})}{\TNat}}\) be a
  \VSystemT{} term.
  \begin{gather*}
    \termstoset{\MaxBoolQuestionInt{\TDialogueTree{t}{\TNat}}}{}\\
    \quad=\AppI{\termstoset{\MaxBoolQuestionIntSym}{}}{\termstoset{\TDialogueTree{t}{\TNat}}{}}
    \tag{by \Cref{def:set-model}}\\
    \quad=\AppI{\termstoset{\MaxBoolQuestionIntSym}{}}{\EncodeExt{\AppI{\DialogueTree}{t}}{\TNat}} \tag{by \Cref{lem:dialogue_agreement}}\\
    \quad={\MaxBoolQuestion{\Prune{\AppI{\DialogueTree}{t}}}} \tag{by \cref{lem:max-bool-question-agreement-ext}}
  \end{gather*}
\end{proof}

Similar to \Cref{lem:main}, the key step in the above proof is the use of
\Cref{lem:dialogue_agreement}, which relies on the logical relation from
\Cref{def:internal-logical-relation}.
With \Cref{lem:main-uni} established, we can now proceed to prove our main
result for functions on the Cantor space: $\ModulusUniIntSym$ computes moduli of
uniform continuity for \VSystemT{}-definable functions on the Cantor space.

\begin{theorem}[\AgdaLink{Theorem-55} Correctness of $\ModulusUniIntSym$]
\label{thm:correctness-of-modulus-of-uniform-continuity}
Let \(t : \TCTerms{\TArrowType{(\TArrowType{\TNat}{\TNat})}{\TNat}}\) be a
\VSystemT{} function on the Baire space.
The result \(\termstoset{\ModulusUniInt{t}}{}\) is a modulus of uniform
continuity of the function \(\termstoset{t}{} \circ \EmbedCBSym\).
\end{theorem}
\begin{proof}
  Let $t : (\TNat \tto \TNat) \tto \TNat$ be a \VSystemT{} term.
  We know by \Cref{thm:correctness_dialogue_tree} that $t$ can be encoded by the inductive dialogue tree
  \(\AppI{\DialogueTree}{t}\).
  We therefore know that, as given by \Cref{lem:mod-uni-dialogue},
  $\ModulusUniSym$ gives a modulus of uniform continuity for
  $\termstoset{t}{}$,
  which is to say that we just have to show
   \(
   \IdTy
   {\termstoset{\MaxBoolQuestionInt{\TDialogueTree{t}}}{}}
   {\MaxBoolQuestion{\Prune{\AppI{\DialogueTree}{t}}}}
   \)
  which is given by~\Cref{lem:main-uni}.
\end{proof}

%%% Local Variables:
%%% mode: latexmk
%%% TeX-master: "../internal-effectful-forcing"
%%% End:

%% Think about this later.
\section{Discussion and conclusion}%
\label{sec:conclusion}

This paper relies on and extends the work on effectful forcing by
Escard\'o~\cite{mhe-effectful-forcing}. We present the first
constructive internalisation of the dialogue trees featured in this technique.
In addition to constructing such trees in \VSystemT, we define the operators
that compute moduli of continuity from dialogue trees inside of \VSystemT.

\emph{Further related work.}
Putting dialogue trees aside, other internalisations of effectful forcing have
been explored before.
In~\cite{Kawai:2019} the authors provide a framework for internalising the
effectful forcing technique for \VSystemT{}.
It abstracts away from dialogue trees to a postulated type with enough structure
to carry out the analogous translation.
Due to \VSystemT{}'s lack of polymorphism, there is a disconnection between
Church encoded dialogue trees and inductive dialogue trees.
Namely, the type of the former is too big and will contain terms that do not
behave at all like dialogue trees.
As a result it seems impossible to apply this framework to the case of Churh
encoded dialogue trees, requiring instead our more direct approach.
Similarly, in~\cite{xu:2020}, with a nucleus as a parameter, a general
translation from \VSystemT into itself is developed. Assuming some
suitable conditions on the parameters of the translation, it is
then proved sound through a logical relation
With different instantiations of this translation the author is able to derive
\VSystemT-definable moduli of continuity, as well as \VSystemT-definable bar
recursion functionals.
The present paper strengthens this continuity result by showing that the
dialogue tree of a \VSystemT function is itself \VSystemT-definable.

In the literature one can also find different extensions of Escard\'o's
effectful forcing.
In~\cite{Sterling:jfp:2021} it is extended to prove that \VSystemT{} validates
the realizable bar thesis, i.e., Brouwer's bar thesis such that the bar is
\VSystemT-definable, which is equivalent to the dialogue continuity.
In~\cite{Baillon+Mahboubi+Pedrot:csl:2022} this technique is further generalised
to handle dependent type theory, and to then show that all functions on the
Baire space of the dependent theory BTT~\cite{Pedrot+Tabareau:lics:2017} are
continuous by building external dialogue trees.
The authors rely on a call-by-name interpretation, which as opposed
to~\cite{mhe-effectful-forcing,Sterling:jfp:2021}, gives rise to a
dialogue-based model of the theory.

In~\cite{Cohen+Paiva+Rahli+Tosun:mfcs:2023} the authors use a different
technique to prove the continuity of the \(\text{TT}^{\Box}_{\mathcal{C}}\)
functions on the Baire and Cantor space. This technique consists of deriving
Brouwer trees internally to the theory using effects.
\(\text{TT}^{\Box}_{\mathcal{C}}\) is an effectful and extensional variant of
MLTT, which therefore allows using effects to build such trees internally to the
theory, using computations introduced in~\cite{Longley:1999}.
While the techniques used
in~\cite{mhe-effectful-forcing,Sterling:jfp:2021,Baillon+Mahboubi+Pedrot:csl:2022}
build dialogue trees by induction on terms, such inductive constructions are not
internalizable without reflection mechanisms, which are not supported by
TT$^{\Box}_{\mathcal{C}}$.
Therefore, \cite{Cohen+Paiva+Rahli+Tosun:mfcs:2023} relies instead on classical
logic to prove finiteness of the computed trees and termination of the internal
program.

The above line of work relies on the dialogue-based notion of continuity defined
in \Cref{defn:continuity} to derive the continuity of functions on the Baire
space and the uniform continuity of functions on the Cantor space.
%
% It was proved using classical logic
% in~\cite[Thm.2.1]{Ghani+Hancock+Pattinson:lmcs:2009}, that the two notions of
% continuity and dialogue continuity are equivalent. It was also proved
% in~\cite[Thm.2]{Capretta+Uustalu:fossacs:2016} that those two notions coincide
% assuming bar induction for a \emph{stable} notion of continuity.
% %
% It therefore remains to clarify whether classical reasoning in the case
% of~\cite{Ghani+Hancock+Pattinson:lmcs:2009} or bar induction in the case
% of~\cite{Capretta+Uustalu:fossacs:2016} are necessary. If so, then dialogue
% continuity would be a strictly stronger notion constructively than standard
% continuity.

%\todo[inline]{do we want to discuss hancock's conjecture that martin metioned?}
% If and when the paper is accepted.

\emph{Open questions and further directions.}
It would interesting to determine whether the assumptions of bar induction and \emph{stable} moduli of continuity used by Capretta and Uustalu~\cite{Capretta+Uustalu:fossacs:2016} to prove the equivalence of continuity with the existence of Brouwer trees are strictly necessary or not.
Finally, we envisage some possible applications of our result, such as characterising the \VSystemT-definable functions $(\Nat \to \Nat) \to \Nat$ in terms of the heights of the trees measured by ordinal notations. For instance, it is natural to conjecture that the heights of such trees are bounded by the ordinal $\epsilon_0$.

\bibliographystyle{plainurl}
\bibliography{references.bib}

\end{document}